%% file: m.tex
\documentclass{LMCS}
\def\doi{9(1:06)2013}
\lmcsheading%
{\doi}
{1--32}
{}
{}
{Mar.~31, 2012}
{Feb.~27, 2013}
{}

\usepackage[utf8x]{inputenc}

\usepackage{enumerate}
\usepackage{hyperref}

\usepackage{tikz}
\usetikzlibrary{automata,fit}
\usetikzlibrary{positioning}
\usetikzlibrary{patterns}

\usepackage{logic9-thm}
 
\newcommand{\tto}{\Rightarrow}
\newcommand{\abs}[0]{\ensuremath{\mathfrak{a}}}

\newcommand{\CC}[0]{\ensuremath{\Phi}} 
\newcommand{\Extra}[0]{\ensuremath{\mathsf{Extra}}}
\newcommand{\ExtraM}[0]{\ensuremath{{\Extra_M}}}
\newcommand{\ExtraMp}[0]{\ensuremath{{\Extra_M^+}}}
\newcommand{\ExtraLU}[0]{\ensuremath{{\Extra_{LU}}}}
\newcommand{\ExtraLUp}[0]{\ensuremath{{\Extra_{LU}^+}}}

\renewcommand{\Nat}[0]{\ensuremath{\mathbb{N}}}
\newcommand{\Integers}[0]{\ensuremath{\mathbb{Z}}}
\newcommand{\reset}[1]{\ensuremath{[#1]}}
\newcommand{\Rpos}[0]{\ensuremath{\mathbb{R}_{\geq 0}}}

\newcommand{\SG}[0]{\ensuremath{ZG}}
\newcommand{\GZG}[0]{\ensuremath{GZG}}
\newcommand{\SGLU}[0]{\ensuremath{\SG^{LU}}}

\newcommand{\RGZG}[0]{\ensuremath{r\GZG}}

\newcommand{\RGZGM}[0]{\ensuremath{\RGZG^M}}
\newcommand{\val}[0]{v} 
\newcommand{\vali}[0]{\ensuremath{\mathbf{0}}} 
\newcommand{\xra}{\xrightarrow}

\newcommand{\zlu}{\SGLU}

\newcommand{\npc}{\NP-complete}
\newcommand{\nph}{\NP-hard}
\newcommand{\ANZ}[0]{\ensuremath{\Aa^{NZ}}}
\newcommand{\AZ}{\ensuremath{\Aa^{Z}}}

\newcommand{\zzg}{\Ss \SG^\abs(\Aa)}
\newcommand{\free}{\operatorname{free}}
\newcommand{\slow}{\operatorname{slow}}

\newcommand{\ZP}{\ensuremath{\mathsf{ZP}}}
\newcommand{\NZP}{\ensuremath{\mathsf{NZP}}}

\newcommand{\Lf}{\operatorname{Lf}}
\newcommand{\U}{\bar{U}}

\newcommand{\elub}[0]{\ensuremath{\Extra_{L \bar{U}}}}
\newcommand{\elbu}[0]{\ensuremath{\Extra_{\bar{L} U}}}
\newcommand{\elubp}[0]{\ensuremath{\Extra^+_{L \bar{U}}}}
\newcommand{\elbup}[0]{\ensuremath{\Extra^+_{\bar{L} U}}}

\newcommand{\Rl}{\operatorname{Rl}}

\newenvironment{definition}[1][]{\begin{defi}[#1]}{\end{defi}}
\newenvironment{theorem}[1][]{\begin{thm}[#1]}{\end{thm}}
\newenvironment{lemma}[1][]{\begin{lem}[#1]}{\end{lem}}

\usepackage[textsize=small]{todonotes}

\begin{document}

\title[Coarse abstractions make Zeno behaviours difficult to detect]{Coarse abstractions make Zeno behaviours difficult to detect\rsuper*}

\author[F.~Herbreteau]{Fr\'ed\'eric Herbreteau\rsuper a}

\address{{\lsuper a}Univ. Bordeaux, CNRS, LaBRI, UMR 5800\\
  F-33400 Talence, France}
\email{fh@labri.fr}

\thanks{{\lsuper a}This work has been supported by ANR project DOTS
  ANR-06-SETI-003.}

\author[B.~Srivathsan]{B.~Srivathsan\rsuper b}

\address{{\lsuper b}Software Modeling and Verification group \\
 RWTH Aachen University \\
 Germany}
\email{sri@cs.rwth-aachen.de}

\thanks{{\lsuper b}The second author B. Srivathsan was at LaBRI, Univ. Bordeaux, when this
  work was first submitted.}



\ACMCCS{[{\bf Software and its engineering}]: Software creation and
  management---Software verification and validation---Formal software
  verification; [{\bf Theory of computation}]: Computational complexity and cryptography---Complexity classes; [{\bf Theory of computation}]: Formal languages and automata theory---Automata extensions}

\titlecomment{{\lsuper*}Extended abstract appeared at CONCUR2011}

\begin{abstract}
  An infinite run of a timed automaton is Zeno if it spans
  only a finite amount of time. Such runs are considered unfeasible
  and hence it is important to detect them, or dually, find runs that
  are non-Zeno.  Over the years important improvements have been
  obtained in checking reachability properties for timed automata. We
  show that some of these very efficient optimizations make testing for
  Zeno runs costly. In particular we show NP-completeness for the
  LU-extrapolation of Behrmann et al. We analyze the source of this
  complexity in detail and give general conditions on extrapolation
  operators that guarantee a (low) polynomial complexity of Zenoness
  checking. We propose a slight weakening of the LU-extrapolation that
  satisfies these conditions.
\end{abstract}

\maketitle

\section*{Introduction}
\label{sec:introduction}

Timed automata~\cite{Alur:TCS:1994} are finite automata augmented with
a finite number of clocks. The values of the clocks increase
synchronously along with time in the states of the automaton and these
values can be compared to a constant and reset to zero while crossing
a transition. This model has been successfully used for verification
of real-time systems thanks to a number of
tools~\cite{behrmann2006uppaal,bozga1998kronos,DBLP:conf/isola/Wang06}.

Since timed automata model reactive systems that continuously interact
with the environment, it is interesting to consider questions related
to their infinite executions. An execution is said to be \emph{Zeno}
if an infinite number of events happen in a finite time interval. Such
executions are clearly unfeasible. During verification, the aim is to
detect if there exists a \emph{non-Zeno} execution that violates a
certain property. On the other hand while implementing timed automata,
it is required to check the presence of pathological Zeno
executions. This brings the motivation to analyze an automaton for the
presence of such executions.

The analysis of timed automata faces the challenge of handling its
uncountably many configurations. To tackle this problem, one considers
a finite graph called the \emph{abstract zone graph} (also known as
simulation graph) of the automaton. This finite graph captures the
semantics of the automaton.  In this paper, we consider the problems
of deciding if an automaton has a non-Zeno execution, dually a Zeno
execution, given its abstract zone graph as input.

An abstract zone graph is obtained by over-approximating each zone of
the so-called \emph{zone graph} with an abstraction function. The zone
graph in principle could be infinite and an abstraction function is
necessary for reducing it to a finite graph. The coarser the
abstraction, the smaller the abstract zone graph, and hence the
quicker the analysis of the automaton. This has motivated a lot of
research towards finding coarser abstraction
functions~\cite{Behrmann:STTT:2006}. The classic maximum-bound
abstraction uses as a parameter the maximal constant a clock gets
compared to in a transition. A coarser abstraction called the
LU-extrapolation was introduced in Behrmann et
al.~\cite{Behrmann:STTT:2006} for checking state reachability in timed
automata. This is the coarsest among all the implemented
approximations and is at present efficiently used in tools like
UPPAAL~\cite{behrmann2006uppaal}.

It was shown in~\cite{Tripakis:TOCL:2009,Tripakis:FMSD:2005} that even
infinite executions of the automaton directly correspond to infinite
paths in the abstract zone graph when one uses the maximum-bound
approximation. In addition, it was proved that the existence of a
non-Zeno infinite execution could be determined by adding an extra
clock to the automaton to keep track of time and analyzing the
abstract zone graph of this transformed
automaton~\cite{Tripakis:ARTS:1999,Tripakis:FMSD:2005}. A similar
correspondence was established in the case of the LU-extrapolation by
Li~\cite{Li:FORMATS:2009}. These results answer our question about
deciding non-Zeno infinite executions of the automaton from its
abstract zone graph. However, it was shown
in~\cite{Herbreteau:FMSD:2012} that adding a clock has an exponential
worst case complexity. A new polynomial construction was proposed for
the case of the classic maximum-bound approximation. But, the case of
the LU-extrapolation was not addressed.

In this paper, we prove that the non-Zenoness question turns out to be
\NP-complete for the LU-extrapolation, that is, given the abstract
zone graph over the LU-extrapolation, deciding if the automaton has a
non-Zeno execution is \NP-complete. We study the source of this
complexity in detail and give conditions on abstraction operators to
ensure a polynomial complexity. To this regard, we extend the
polynomial construction given in~\cite{Herbreteau:FMSD:2012} to an
arbitrary abstraction function and analyze when it stays
polynomial. It then follows that a slight weakening of the
LU-extrapolation makes the construction polynomial. In the second part
of the paper, we repeat the same for the dual question: given an
automaton's abstract zone graph, decide if it has Zeno executions.
Yet again, we notice \NP-completeness for the LU-extrapolation. We
introduce an algorithm for checking Zenoness over an abstract zone
graph with conditions on the abstraction operator to ensure a
polynomial complexity. We provide a different weakening of
LU-extrapolation that gives a polynomial solution to the Zenoness
question. Finally, we also prove that deciding if a given automaton
has a non-Zeno run (resp. Zeno run) is \PSPACE-complete when the input
is restricted to the automaton only.

Note that the reachability problem for timed automata is
\PSPACE-complete \cite{Alur:TCS:1994,Courcoubetis:FMSD:1992} and the
standard algorithms make use of the abstract zone graph to solve the
reachability problem. Therefore one could expect an object as complex
as the abstract zone graph to solve the Zeno-related questions
too. This makes the complexity analysis of the Zeno-related questions,
given both the automaton and abstract zone graph as input, all the
more relevant.
 
\subsection*{Related work}

As mentioned above, the LU-extrapolation was proposed
in~\cite{Behrmann:STTT:2006} and shown how it could be efficiently
used in UPPAAL for the purpose of reachability. The correctness of the
classic maximum-bound abstraction was shown
in~\cite{Bouyer:FMSD:2004}. Extensions of these results to infinite
executions occur
in~\cite{Tripakis:FMSD:2005,Li:FORMATS:2009}. Detection of non-Zeno
runs was already addressed in~\cite{Alur:TCS:1994}. Their approach
works on the region graph, but for correctness reasons, it cannot be used on (abstract)
zone graphs. The trick involving adding an extra clock for
non-Zenoness is discussed
in~\cite{Tripakis:ARTS:1999,Tripakis:FMSD:2005,Alur:SFM:2004,Herbreteau:FMSD:2012}.
The problem of checking existence of Zeno runs was formulated as early
as in~\cite{Tripakis:ARTS:1999}. A bulk
of the literature for this problem also directs
to~\cite{gomez2007efficient,bowman2006stop,Rinast:FORMATS:2012}. All
of these solutions
provide a sufficient-only condition for the absence of Zeno runs. This
is different from our proposed solution which gives a complete
solution (necessary and sufficient conditions) by analyzing the
abstract zone graph of the automaton.

\subsection*{Organization of the paper}

We start with the formal definitions of timed automata, abstract zone
graphs, the Zenoness and non-Zenoness problems in
Section~\ref{sec:ta-zeno}. Subsequently, we prove the \NP-hardness
of the non-Zenoness problem for the LU-extrapolation in
Section~\ref{sec:nz-npc-lu}. We then recall in
Section~\ref{sec:gzg-lu} the construction proposed
for non-Zenoness in~\cite{Herbreteau:FMSD:2012} and extend it to a
general abstraction operator giving conditions for polynomial
complexity. Section~\ref{sec:zenoness-problem} is dedicated to the
dual Zenoness problem. In Section~\ref{sec:discussion} we discuss some
interesting observations arising out of the entire complexity
analysis. We prove in Section~\ref{sec:pspace-complete-zeno} that
finding if an automaton has a (non-)Zeno run turns out to
\PSPACE-complete when the input is restricted to the automaton
only. This gives a complete characterization of the complexity of
finding (non-)Zeno runs in timed automata. We conclude the
paper with some perspectives in Section~\ref{sec:conclusion}.

\medskip

A shorter version of this paper appeared at the $22^{nd}$
International Conference on Concurrency Theory in the year
2011~\cite{Herbreteau:CONCUR:2011}. The current version includes the missing
proofs, a new discussion (Section~\ref{sec:discussion}) about two
observations arising out of the complexity analysis, and the new
result about the \PSPACE-completeness of the Zeno-related problems
when the only input is the automaton
(Section~\ref{sec:pspace-complete-zeno}).

\section{Zeno-related Problems for Timed Automata}
\label{sec:ta-zeno}

\subsection{Timed automata}
\label{sec:ta}

Let $\Rpos$ denote the set of non-negative real numbers. Let $X$ be a
set of variables, named \emph{clocks} hereafter. A \emph{valuation} is
a function $\val:X \mapsto \Rpos$ that maps every clock in $X$ to a
non-negative real value. We denote the set of all valuations by
$\Rpos^X$, and $\vali$ the valuation that maps every clock in $X$ to
$0$. For $\d\in\Rpos$, we denote $\val+\d$ the valuation mapping each
$x\in X$ to the value $\val(x)+\d$. For a subset $R$ of $X$, let
$\reset{R}\val$ be the valuation that sets $x$ to $0$ if $x\in R$ and
assigns $\val(x)$ otherwise. A \emph{clock constraint} is a
conjunction of constraints $x\# c$ for $x\in X$,
$\#\in\{<,\leq,=,\geq,>\}$ and $c\in \Nat$. We denote $\CC(X)$ the set
of clock constraints over clock variables $X$. For a valuation $\val$
and a constraint $\phi$ we write $\val \sat \phi$ when $\val$
satisfies $\phi$, that is, when $\phi$ holds after replacing every $x$
by $\val(x)$.

A \emph{Timed Automaton (TA)}~\cite{Alur:TCS:1994} $\Aa$ is a finite
automaton extended with clocks that enable or disable
transitions. Formally, $\Aa$ is a tuple $(Q,q_0,X,T)$ where $Q$ is a
finite set of states, $q_0\in Q$ is the initial state, $X$ is a finite
set of clocks and $T\subseteq Q\times\CC(X)\times 2^X\times Q$ is a
finite set of transitions. For each transition $(q,g,R,q')\in T$, $g$
is a clock constraint, also called a \emph{guard}
 that defines the valuations of the clocks that are allowed
to cross the transition, and $R$ is a set of clocks that are \emph{reset} on
the transition.

The semantics of a timed automaton $\Aa$ is a transition system of its
configurations.  A \emph{configuration} of $\Aa$ is a pair
$(q,\val)\in Q\times\Rpos^X$, with $(q_0,\vali)$ being the
\emph{initial configuration}. We have two kinds of transitions:
\smallskip
\begin{description}
\item[delay] $(q,\val)\to^\d(q,\val+\d)$ for some $\d\in \Rpos$;
  \smallskip
\item[action] $(q,\val)\to^t(q',\val')$ for some transition $t =
  (q,g,R,q')\in T$ such that $\val \sat g$ and $\val'=[R]\val$.
\end{description}
\smallskip

A \emph{run} of $\Aa$ is a (finite or infinite) sequence of
transitions starting from the initial configuration $(q_0,
\vali)$. Without loss of generality, we can assume that the first
transtition is a delay transition and that delay and action
transitions alternate. We write $(q, \val) \xra{\d, t} (q',\val')$ if
there is a delay transition $(q,\val) \to^\d (q, \val+\d)$ followed by
an action transition $(q, \val+\d) \to^{t} (q',\val')$.  So a run of
$\Aa$ can be written as:
\begin{equation*}
  (q_0, \val_0) \xra{\d_0, t_0} 
  (q_1, \val_1) \xra{\d_1, t_1} 
  (q_2, \val_2) \cdots 
  (q_i, \val_i) \cdots
\end{equation*}
where $(q_0, \val_0)$ represents the initial configuration $(q_0,
\vali)$.

\begin{definition}[Zeno/non-Zeno runs]
  \label{defn:non-zeno-run}
  An infinite run $(q_0,v_0)\xra{\d_0,t_0}\dots
  (q_i,\val_i)\xra{\d_i,t_i}\dots$ is \emph{Zeno} if time
  does not diverge, that is, $\sum_{i\geq 0} \d_i \le c$ for some $c
  \in \Rpos$. Otherwise it is
  \emph{non-Zeno}. 
\end{definition}

\begin{thm}
  \label{thm:Z-NZ-PSPACE}
  The problem of deciding if a timed automaton $\Aa$ has a non-Zeno
  run (resp. Zeno run) is \PSPACE-complete if $\Aa$ is the only input.
\end{thm}

A proof of Theorem~\ref{thm:Z-NZ-PSPACE} is given in
Section~\ref{sec:pspace-complete-zeno}
(page~\pageref{sec:pspace-complete-zeno}) that relies on results in
Sections~\ref{sec:gzg-lu} and~\ref{sec:zeno-P-algorithm}.

As can be seen, the number of configurations $(q, \val)$ could be
uncountable. We now define an abstract semantics for timed
automata. The abstract semantics is usually used for the verification
of timed automata. 

\subsection{Symbolic semantics}
\label{sec:ss-ta}

We begin with the definition of special \emph{sets of valuations}
called zones. A \emph{zone} is a set of valuations defined by a
conjunction of two kinds of clock constraints: for $x_i, x_j \in X$
\begin{align*}
  x_i \sim c \\
  x_i - x_j \sim c
\end{align*}
where $\sim \in \{ \le, <, = , >, \ge\}$ and $c \in \Integers$. An
example of a zone over two clocks $x_1$ and $x_2$ is illustrated in
Figure~\ref{fig:zone-example}. The shaded area is the zone represented
by the conjunction of the six constraints shown in the figure.

\begin{figure}[t]
  \centering \input{fig/zone-boundaries}
  \caption{An example of a zone.}
  \label{fig:zone-example}
\end{figure}
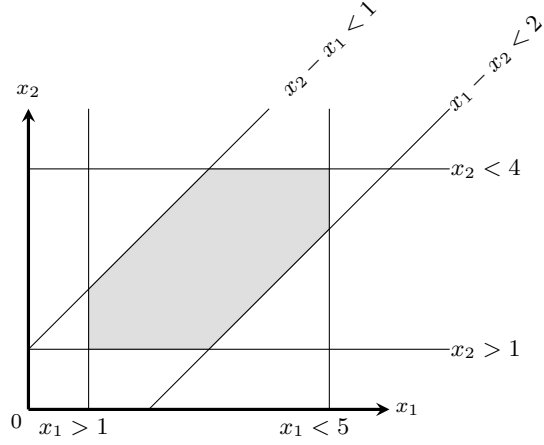

Zones can be efficiently represented by Difference Bound Matrices
(DBMs)~\cite{Dill:AVMFSS:1989}. A DBM representation of a zone $Z$ is
a $|X|+1$ square matrix $(Z_{ij})_{i,j\in[0;|X|]}$ where each entry
$Z_{ij}=(c_{ij},\preccurlyeq_{ij})$ represents the constraint
$x_i-x_j\preccurlyeq_{ij} c_{ij}$ for $c_{ij}\in\Integers$ and
$\preccurlyeq_{ij}\in\{<,\leq\}$ or $(c_{ij}, \preccurlyeq_{ij}) =
(\infty, <)$. A special variable $x_0$ encodes the value $0$. Hence,
in a DBM $x_i > 4$ is encoded as $x_0 - x_i < -4$.

The \emph{symbolic semantics} (or \emph{zone graph}) of an automaton
$\Aa$ is the transition system $\SG(\Aa)$ given by the tuple $(S, s_0,
\tto)$, where $S$ is the set of nodes, $s_0$ is the initial node and
$\tto$ is the transition relation. Each node in $S$ is a pair $(q, Z)$
consisting of a state $q$ of the automaton and a zone $Z$. The initial
node $s_0$ is $(q_0, Z_0)$ where $Z_0 = \{ \vali+ \d~|~ \d \in
\Rpos\}$. For every $t = (q, g, R, q') \in T$, there exists a
transition $\tto^t$ from a node $(q,Z)$ as follows:
\begin{align*}
  (q, Z) \tto^t (q', Z') \qquad \text{where } Z' =\set{ \val'~|~
    \exists \val \in Z, ~\exists \d \in \Rpos: (q,\val) \to^t\to^\d
    (q',\val')}
\end{align*}

In the above definition, $\to^t\to^\d$ denotes the discrete transition
$t$ followed by a delay transition of $\d$ time units.  It can be
shown that if $Z$ is a zone, then $Z'$ is a zone. Moreover, a DBM
representation of $Z'$ can be computed from the DBM representation of
$Z$ (see for
instance~\cite{Bouyer:FMSD:2004}). Figure~\ref{fig:zg-example} shows
an example of an automaton and its zone graph.

\begin{figure}[t]
  \centering \input{fig/ta-example} \vspace{0.4cm}
  \input{fig/zg-example}
  \caption{A timed automaton (top) and its zone graph (bottom).}
  \label{fig:zg-example}
\end{figure}
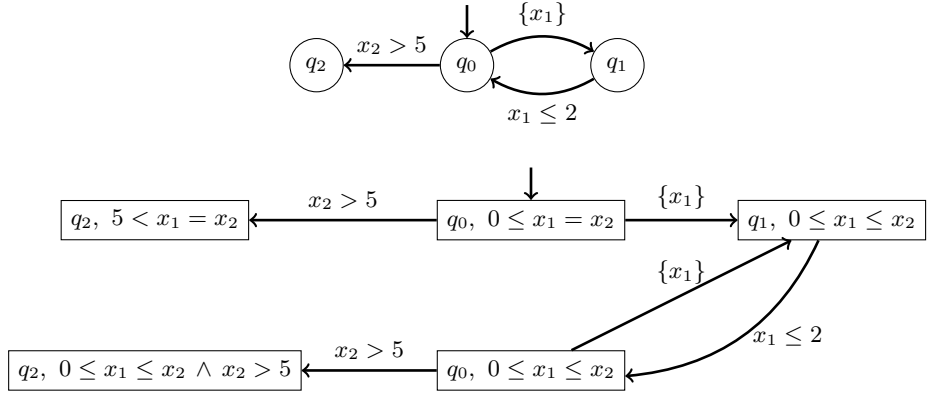

Several definitions of $\SG(\Aa)$ have been considered in the
literature. They differ on the definition of $\tto$. People have
considered graphs with both action $\tto^t$ and delay $\tto^\d$
transitions or, with only combined transitions $\tto^t$, but
corresponding to the reverse consecution $\to^\d\to^t$
(delay-then-action). Our results do not depend on a specific choice,
but have a simpler presentation using the chosen symbolic semantics.

Although the zone graph $\SG(\Aa)$ deals with sets of valuations
instead of valuations themselves, the zone graph could stil be
infinite.  Consider the automaton $\Aa_{inf}$ shown in
Figure~\ref{fig:a-inf}, with two clocks $x_1$ and $x_2$. The initial
node is given by $(q_0, x_1 = x_2 ~\land~ x_1 \ge 0)$. The transition
to $q_1$ gives the node $(q_1, x_1 = x_2 ~\land~ x_1 \ge 0)$. The only
transition from $q_1$ taken from this node gives the node $(q_1, x_1 -
x_2 = 1 ~\land~ x_1 \ge 0)$, which is a new node. This node has its
own successors and the process continues. Finally at $q_1$ we have the
following zones in the zone graph $\SG(\Aa_{inf})$:
\begin{equation*}
  (~x_1 - x_2 = k ~\land~ x_1 \ge 0~) ~ \text{ for all } k \in \Nat 
\end{equation*}
This is pictorially shown in Figure~\ref{fig:a-inf}. It is however
sufficient to consider a finite abstraction of the zone graph to
capture all the behaviors of a timed automaton. Several abstractions
have been introduced to obtain a finite graph from $\SG(\Aa)$.

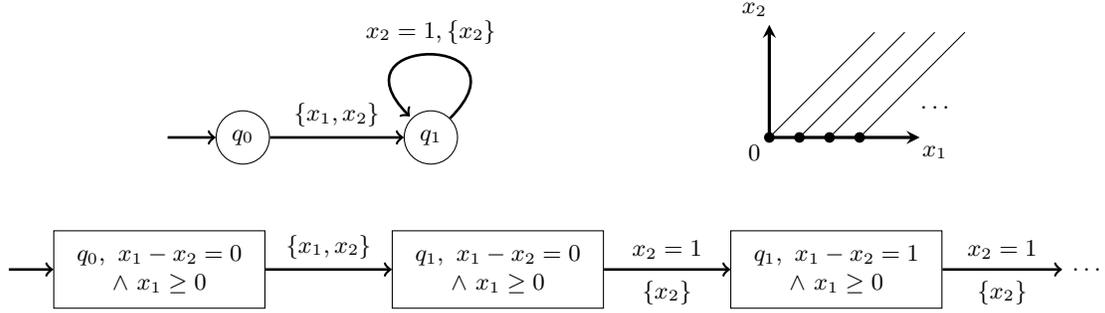
\begin{figure}[t]
  \input{fig/zg-ainf}
  \caption{Automaton $\Aa_{inf}$ (top left), the graph of zones
    obtained at $q_1$ (top right) and a part of the infinite zone
    graph $\SG(\Aa_{inf})$ (bottom).}
  \label{fig:a-inf}
\end{figure}

\subsection{Abstract symbolic semantics}
A \emph{finite abstraction} is a function $\abs: 2^{\Rpos^{|X|}} \to
2^{\Rpos^{|X|}}$ such that for every zone $Z$: $\abs(Z)$ is a zone,
$Z\incl\abs(Z)$, $\abs(\abs(Z))=\abs(Z)$ and $\abs$ has a finite
range. An abstraction operator defines an abstract semantics.
\begin{align*}
  (q,Z) \tto^t_{\abs }(q',\abs(Z'))
\end{align*}
when $\abs(Z)=Z$ and $(q,Z)\tto^t (q',Z')$ in $\SG(\Aa)$. 

The \emph{abstract symbolic semantics} (or the \emph{abstract zone
  graph}) of $\Aa$ is the transition system $\SG^\abs(\Aa)$ induced by
$\tto_{\abs}$ starting from the node $(q_0, \abs(Z_0))$, where
$(q_0,Z_0)$ is the initial node of $\SG(\Aa)$.

A \emph{path} $\pi$ in $\SG^\abs(\Aa)$ is a (finite or infinite)
sequence of transitions
\begin{align*}
  (q_0, Z_0') \tto^{t_0}_\abs 
  (q_1, Z_1') \tto^{t_1}_\abs \cdots 
  (q_i, Z_i') \tto^{t_i}_\abs \cdots
\end{align*}

We say that a run $\rho$: $(q_0,v_0)\xra{\d_0,t_0}\dots
(q_i,\val_i)\xra{\d_i,t_i}\dots$ of $\Aa$ is an \emph{instance} of the
path $\pi$ of $\SG^\abs(\Aa)$ as described above, if $\rho$ and $\pi$
agree on the sequence of transitions $t_0, t_1, \dots$, and if for
every $i\geq 0$, $(q_i,\val_i)$ and $(q_i,Z_i')$ coincide on $q_i$,
and $\val_i\in Z_i'$. By definition of $Z_i'$, this implies
$\val_i+\d_i\in Z'_i$.

An abstraction $\abs$ is \emph{sound} if every path of $\SG^\abs(\Aa)$
can be instantiated as a run of $\Aa$. Conversely, $\abs$ is
\emph{complete} when every run of $\Aa$ is an instance of some path in
$\SG^\abs(\Aa)$. If an abstraction $\mathfrak{b}$ satisfies
$\mathfrak{b}(Z) \incl \abs(Z)$ for every $Z$, it is easy to see that
the abstract zone graph $\SG^{\mathfrak{b}}(\Aa)$ is bigger than 
$\SG^\abs(\Aa)$.

\subsection{Bounds and finite abstractions}

A standard way to obtain finite abstractions is to consider as a
parameter, a bound function $M : X \mapsto \Nat \cup \{-\infty\}$ that
associates to each clock $x$, the maximum integer $c$ appearing in a
guard involving $x$. Abstractions $\ExtraM$~\cite{Bouyer:FMSD:2004}
and $\ExtraMp$~\cite{Behrmann:STTT:2006} are well-known finite
abstractions that depend on such a bound function $M$.

It has been observed that considering separately the guards that
lower-bound clocks and guards that upper-bound clocks leads to much
coarser abstractions and hence to much smaller abstract zone
graphs. This has given rise to abstractions $\ExtraLU$ and
$\ExtraLUp$~\cite{Behrmann:STTT:2006} which are currently used in
implementations. We recall the definitions of $\ExtraLU$ and
$\ExtraLUp$ below.

Let $L: X\mapsto \Nat\cup\{-\infty\}$ and $U: X\mapsto
\Nat\cup\{-\infty\}$ be two maps that associate to each clock in $\Aa$
its maximal lower bound and its maximal upper bound respectively: that
is, for every $x\in X$, $L(x)$ is the maximal integer $c$ such that
$x>c$ or $x\geq c$ appears in some guard of $\Aa$. We let
$L(x)=-\infty$ if no such $c$ exists. Similarly, we define $U(x)$ with
respect to clock constraints like $x\leq c$ and $x<c$. We define
$\ExtraLU(Z)=Z^{LU}$ and $\ExtraLUp(Z)=Z^{LU+}$ as:

{\footnotesize\noindent
  \begin{minipage}{.465\textwidth}
    \begin{equation*}
      Z_{ij}^{LU}=
      \begin{cases}
        (\infty,<) & \text{if}\ c_{ij}>L(x_i)\\
        (-U(x_j),<) & \text{if}\ -c_{ij}>U(x_j)\\
        Z_{ij} & \text{otherwise}
      \end{cases}
    \end{equation*}
  \end{minipage}
  \vrule
  \begin{minipage}{.53\textwidth}
    \begin{equation*}
      Z_{ij}^{LU+}=
      \begin{cases}
        (\infty,<) & \text{if}\ c_{ij}>L(x_i)\\
        (\infty,<) & \text{if}\ -c_{0i}>L(x_i)\\
        (\infty,<) & \text{if}\ -c_{0j}>U(x_j),i\neq 0\\
        (-U(x_j),<) & \text{if}\ -c_{0j}>U(x_j),i=0\\
        Z_{ij} & \text{otherwise}
      \end{cases}
    \end{equation*}
  \end{minipage}
} 

In the above, we set $L(x_0)=U(x_0)=0$ for the special clock
$x_0$. The abstraction $\ExtraM$ (resp. $\ExtraMp$) is obtained from
$\Extra_{LU}$ (resp. $\ExtraLUp$) by replacing every occurrence of $L$
and $U$ by $M$ which maps every clock $x$ to $\max(L(x),U(x))$. These
abstractions compare in the following way
(cf. Figure~\ref{fig:abs_lit}).

\begin{theorem}[\cite{Behrmann:STTT:2006}]
  \label{thm:inclusion-abstractions}
  For every zone $Z$, we have: $Z\subseteq \ExtraM(Z)\subseteq
  \ExtraMp(Z)$; $Z \subseteq \ExtraLU(Z)\subseteq \ExtraLUp(Z)$ and
  $\ExtraMp(Z) \subseteq \ExtraLUp(Z)$.
\end{theorem}

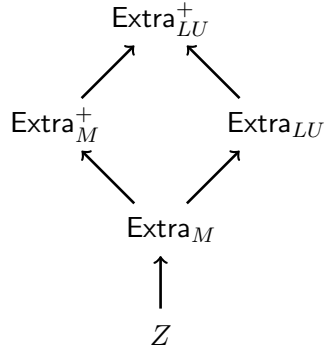
\begin{figure}[!t]
  \centering
  \begin{tikzpicture}[scale=0.7,line width=1pt]
    \draw (2,-2) node {$Z$} (4,-2); \draw (0,2) node {$\ExtraMp$}
    (1,2); \draw (4.2,2) node {$\ExtraLU$ } (1,2); \draw (2.2,0) node
    {$\ExtraM$ } (4,0); \draw (2,4) node {$\ExtraLUp$} (3,4);

    \draw[->] (2, -1.5) -- (2,-0.5); \draw[->] (1.5,0.5) -- (0.5,1.5);
    \draw[->] (0.5, 2.5) -- (1.5,3.5); \draw[->] (2.5, 0.5) -- (3.5,
    1.5); \draw[->] (3.5, 2.5) -- (2.5, 3.5);
  \end{tikzpicture}
  \caption{Comparison of the finite
    abstractions~\cite{Behrmann:STTT:2006}.}
  \label{fig:abs_lit}
\end{figure}

Figure~\ref{fig:abst-examples} shows a zone and depicts the action of
the different abstractions on it. In the rest of the paper, we say
$M$-extrapolations for $\ExtraM$ and $\ExtraMp$; and
$LU$-extrapolations for $\ExtraLU$ and $\ExtraLUp$.

\input{fig/abs-example}
 
Let us look at the timed automaton $\Aa_{inf}$ of
Figure~\ref{fig:a-inf}. For this automaton, the maximum bounds
function $M$ sets $M(x_1) = -\infty$ and $M(x_2) = 1$. Define:
\begin{align*}
\text{for } k \in \Nat, \quad Z_k \equiv (x_1 - x_2 = k) \ \land \
(x_1 \ge 0)
\end{align*}
By the definition of $\ExtraM$, every zone $Z_k$ has $\ExtraM(Z_k)$
given by the constraints $(x_1 \ge 0 \land x_2 \ge 0$). Therefore, the
zone graph $\SG^\abs(\Aa_{inf})$ has two nodes for $\abs$ being any of
the abstractions defined above.

\begin{figure}[t]
  \input{fig/abszg-example}
  \caption{$\SG^\abs(\Aa_{inf})$ for the automaton $\Aa_{inf}$ shown in
    Figure~\ref{fig:a-inf}. We get the same abstract zone graph
    $\SG^\abs(\Aa_{inf})$ for $\abs$ being either $\ExtraM, \ExtraMp,
    \ExtraLU$ or $\ExtraLUp$.}
  \label{fig:abszg-eg}
\end{figure}

\subsection{Zenoness and non-Zenoness problems}

A classical verification problem for timed automata is to answer state
reachability queries. For this purpose, we consider only the runs of
$\Aa$ and paths in $\SG^\abs(\Aa)$ that are \emph{finite} sequences of
transitions. A reachability query asks if there exists a run of $\Aa$
leading to a given state. Such problems can be solved using
$\SG^\abs(\Aa)$ when $\abs$ is sound and complete. This is true for
the $M$-extrapolations and $LU$-extrapolations.

\begin{theorem}[\cite{Bouyer:FMSD:2004,Behrmann:STTT:2006}]
  \label{thm:extraLU-extraM-sound-complete-finite}
  $\ExtraM$, $\ExtraMp$, $\ExtraLU$ and $\ExtraLUp$ are sound and
  complete for finite sequences of transitions.
\end{theorem}

Liveness properties require the existence of an infinite run
satisfying a given property. For instance, does $\Aa$ visit state $q$
infinitely often? Soundness and completeness of $\abs$ with respect to
infinite runs allow to solve such problems from $\SG^\abs(\Aa)$. It
has also been proved that the extrapolations mentioned above are also
sound and complete for infinite paths/runs.

\begin{theorem}[\cite{Tripakis:TOCL:2009,Li:FORMATS:2009}]
  \label{thm:extraLU-extraM-sound-complete-infinite}
  $\ExtraM$, $\ExtraMp$, $\ExtraLU$ and $\ExtraLUp$ are sound and
  complete for infinite sequences of transitions.
\end{theorem}

Thanks to Theorem~\ref{thm:extraLU-extraM-sound-complete-infinite}, we
know that every infinite path $\pi$ in $\SG^\abs(\Aa)$ can be
instantiated to a run of $\Aa$. However, soundness is not sufficient
to know if $\pi$ can be instantiated to a \emph{non-Zeno}
run. Additionally, it is also interesting to know when this path can
be instantiated to a \emph{Zeno} run.  In the sequel, we consider the
following questions, given an automaton $\Aa$ and an abstract zone
graph $ZG^\abs(\Aa)$.

\bigskip
\begin{tabular}{l l l}
  \textsc{Input} & \hspace{1em} & $\Aa$ and $\SG^\abs(\Aa)$ \\
  \textsc{Non-Zenoness problem}~($\NZP^\abs$) & & Does $\Aa$ have a
  non-Zeno run?\\
  \textsc{Zenoness problem}~($\ZP^\abs$) & & Does $\Aa$ have a Zeno
  run?\\
\end{tabular}

\bigskip

\noindent Observe that solving $\ZP^\abs$ does not solve $\NZP^\abs$ and
vice-versa: one is not the negation of the other. Note that the
coarser the abstraction, the lesser is the information maintained
about the structure of a zone. Let us motivate by an example. 

\begin{figure}[!t]
  \centering \input{fig/zeno-abstract-paths}
  \caption{Zenoness/non-Zenoness from abstract paths.}
  \label{fig:zeno-abstract-paths}
\end{figure}

Figure~\ref{fig:zeno-abstract-paths} shows an automaton $\Aa_{zeno}$
which has all runs Zeno. As we can see, the coarser the abstraction
used, the lesser is the information in the simulation graph that one
could tap to detect non-Zenoness or Zenoness.

In this paper, we focus on the complexity of deciding $\ZP^\abs$ and
$\NZP^\abs$ for different abstractions $\abs$. We denote $\NZP^M$ and
$\ZP^M$ when the $M$-extrapolations are considered. We similarly
define $\NZP^{LU}$ and $\ZP^{LU}$ for the $LU$-extrapolations.

The non-Zenoness problem is known to be solvable in polynomial time
when abstraction $\ExtraM$ is
considered~\cite{Herbreteau:FMSD:2012}. This is not true for
abstraction $\ExtraLU$: in Section~\ref{sec:nz-npc-lu} we show that
$\NZP^{LU}$ is \NP-hard. As the $LU$-extrapolations are coarser, one
might expect the non-Zenoness question to be tougher to infer. But, it
is surprising that the difficulty rises to the extent of leading to
NP-hardness as opposed to a low polynomial complexity for
$M$-extrapolations. The same asymmetry appears in the Zenoness problem
as well, which is shown in Section~\ref{sec:zenoness-problem}.

In addition to these complexity results, in Section~\ref{sec:gzg-lu},
we generalize the construction for non-Zenoness given
in~\cite{Herbreteau:FMSD:2012} to an arbitrary finite abstraction
operator and describe the class of abstractions for which $\NZP^\abs$
stays polynomial. The $M$-extrapolations satisfy this criteria. We
show that a small weakening of the $LU$-extrapolations that preserves
an ordering property between clocks also satisfies this criterion. In
Section~\ref{sec:zenoness-problem}, we give an algorithm for
$\ZP^\abs$ and describe the class of abstractions that give a
polynomial complexity. Yet again, the $M$-extrapolations satisfy these
criteria. We will see that a weakening of the $LU$-extrapolations that
maintains some lower-bound information also satisfies this criterion.

\section{Non-Zenoness is \nph{} for \texorpdfstring{$LU$}{LU}-extrapolations}
\label{sec:nz-npc-lu}

We give a reduction from the 3SAT problem: given a 3CNF formula
$\phi$, we build an automaton $\ANZ_\phi$ that has a non-Zeno run iff
$\phi$ is satisfiable. The size of the automaton will be linear in the
size of $\phi$. We will then show that the abstract zone graph
$\SGLU(\ANZ_\phi)$ is isomorphic to $\ANZ_\phi$, thus
completing the polynomial reduction from 3SAT to $\NZP^{LU}$.

\subsection*{Automaton \texorpdfstring{$\ANZ_\phi$}{ANZphi}}

Let $P=\{p_1,\dots,p_k\}$ be a set of propositional variables and let
$\phi = C_1 \land \dots \land C_n$ be a 3CNF formula with $n$ clauses.
We define the timed automaton $\ANZ_\phi$ as follows. Its set of
clocks $X$ equals $\{x_1,\bar{x_1},\dots,x_k,\bar{x_k}\}$. For a
literal $\l$, let $cl(\l)$ denote the clock $x_i$ when $\l=p_i$ and
the clock $\bar{x_i}$ when $\l=\neg p_i$. The set of states of
$\ANZ_\phi$ is $\{q_0,\dots,q_k,r_0,\dots,r_n\}$ with $q_0$ being the
initial state. The transitions are as follows:
\begin{iteMize}{$\bullet$}
\item for each proposition $p_i$ we have transitions
  $q_{i-1}\xra{\{x_i\}} q_i$ and $q_{i-1}\xra{\{\bar{x_i}\}} q_i$,
\item for each clause $C_m=\l_1^m\vee \l_2^m\vee \l_3^m$, $m = 1,
  \dots, n$, there are three transitions $r_{m-1}\xra{cl(\l) \le 0}
  r_m$ for $\l \in \{\l_1^m, \l_2^m, \l_3^m\}$,
\item transitions $q_k\xra{} r_0$ and $r_n\xra{} q_0$ with no guards
  and resets.
\end{iteMize}

\input{fig/nonzeno.tex}

Figure~\ref{fig:nonzeno-eg} shows the automaton for the formula $(p_1
\vee \neg p_2 \vee p_3) \wedge (\neg p_1 \vee p_2 \vee p_3)$. The part
from $q_0$ to $q_3$ encodes an assignment with the following
convention: a reset of $x_i$ represents $p_i\mapsto true$ and a reset
of $\bar{x_i}$ means $p_i\mapsto false$.  Then, from $r_0$ to $r_2$ we
check if the formula is satisfied by this guessed assignment.

The above formula is satisfied by every assignment that maps $p_3$ to
$true$. Any path that encodes this assignment using the convention
mentioned above should pick
the transition $q_2\xra{\{x_3\}} q_3$. Then, it has the possibility to
follow transitions $r_0\xra{x_3\leq 0} r_1$ and $r_1\xra{x_3\leq 0}
r_2$. On any cycle containing these three transition, time can elapse
(for instance in state $q_0$) since $x_3$ is reset before being
checked for zero. Therefore, this assignment that makes the formula
true corresponds to a non-Zeno run of $\ANZ_\phi$.

Conversely, consider the assignment $p_1\mapsto
false$, $p_2\mapsto true$ and $p_3\mapsto false$ that does not satisfy
the formula. Take a cycle that resets $\bar{x_1}$, $x_2$ and
$\bar{x_3}$ according to the encoding of assignments. Then none of the
clocks that are checked for zero on the transitions from $r_0$ to
$r_1$ has been reset. Notice that these transitions come from the
first clause in the formula that evaluates to $false$ according to the
assignment. To take a transition from $r_0$, one of $x_1$, $\bar{x_2}$
and $x_3$ must be zero and hence time cannot elapse in the path
corresponding to this assignment.

Lemma~\ref{lem:red-3SAT-non-Zenoness} below states that if the formula
is satisfiable, there exists a sequence of resets that allows time
elapse in every loop. Conversely, if the formula is unsatisfiable, in
every iteration of the loop, there is a zero-check that prevents time
from elapsing.

\begin{lemma}
  \label{lem:red-3SAT-non-Zenoness}
  A 3CNF formula $\phi$ is satisfiable iff $\ANZ_\phi$ has a non-Zeno
  run.
\end{lemma}

\begin{proof}
  Let $\phi$ be a conjunction of $n$ clauses $C_1, \dots, C_n$.
  Assume that $\phi$ is satisfiable. Then, there exists a variable
  assignment $\chi: P \mapsto \{true, false\}$ that evaluates $\phi$
  to true. This entails that in every clause $C_m$ there is a literal
  $\l_m$ that evaluates to true with $\chi$.

  We will now build a non-Zeno run $\r$ of $\ANZ_\phi$ using this
  variable assignment $\chi$.  Clearly, it should have
  the following sequence of states repeated infinitely often:

  \begin{equation*}
    q_0 \xra{} \dots q_k \xra{} r_0 \xra{} r_1 \xra{} \dots r_n
  \end{equation*}

  Additionally, $\r$ satisfies the following conditions:

  \begin{iteMize}{$\bullet$}
  \item from each configuration $(q_{i-1},\val)$ for $i\in[1;k]$,
    $\rho$ takes the transition $q_{i-1}\xra{\{x_i\}} q_i$ when
    $\chi(p_i)=true$ and the transition $q_{i-1}\xra{\{\bar{x_i}\}}
    q_i$ otherwise;
  \item from each configuration $(r_{m-1},\val)$ for $m\in[1;n]$,
    $\rho$ takes a transition $r_{m-1}\xra{cl(\l_m)\leq 0} r_m$ where
    $\l_m$ is the literal evaluating to $true$ with respect to $\chi$
    in $C_m$;
  \item and $\r$ lets $1$ time unit elapse from each configuration
    with state $r_n$ and moves to the state $q_0$; in all other
    states, there is no time elapse.\smallskip
  \end{iteMize}
  
  \noindent Note that as $r_n$ occurs infinitely often, the run $\r$ is
  non-Zeno. It remains to prove that $\r$ is indeed a valid run of
  $\ANZ_\phi$. For this, we need to prove that all zero-checked
  transitions can be crossed regardless of the unit time
  elapse. Consider the part of $\r$ between two successive
  configurations with state $r_n$.
  \begin{equation*}
    \cdots
    (r_n,\val)\xra{1}
    \cdots
    \xra{\{cl(\l_m)\}}
    \cdots
    (r_{m-1},\val'')\xra{cl(\l_m)\leq 0} (r_m,\val'')
    \cdots
    (r_n,\val')\xra{1}
    \cdots
  \end{equation*}
  By definition of $\r$, $\l_m$ is a literal that evaluates to $true$
  according to $\chi$. Hence, clock $cl(\l_m)$ is reset in the
  corresponding $q_{j-1} \xra{} q_{j}$ transition, before being
  checked for zero.  As $cl(\l_m)$ is reset and since $\r$ does not
  elapse time in states other than $r_n$, we have
  $\val''(cl(\l_j^m))=0$. This permits the transition from $r_{m-1}$
  to $r_m$ for all $m \in [1;n]$ and shows that the run $\r$ exists.

  \smallskip

  For the other direction, consider a non-Zeno run $\r$ of
  $\ANZ_\phi$. Since $\r$ is non-Zeno, time elapses on infinitely many
  transitions in the run. Every infinite run of $\ANZ_\phi$ visits a
  configuration with state $r_n$ infinitely often. Consider two
  consecutive occurences of $r_n$ in $\r$ such that time elapses on
  some transition in the segment in between:
  \begin{equation*}
    \cdots 
    (r_n,\val)\xra{}
    \cdots
    (q_k,\val')\xra{}
    \cdots
    (r_{m-1},\val'')\xra{cl(\l_m)\le 0}
    (r_m,\val'')
    \cdots
    \xra{}
    (r_n,\val'')
    \cdots
  \end{equation*}
  By construction, for each $i\in[1;k]$ either $x_i$ or $\bar{x_i}$ is
  reset on the segment from $(r_n,\val)$ to $(q_k,\val')$. Let $\chi$
  be the variable assignment that associates $true$ to $p_i$ when
  $x_i$ is reset, and $false$ otherwise, that is when $\bar{x_i}$ is
  reset. We prove that $\chi$ satisfies $\phi$.

  Consider the transition $(r_{m-1},\val'')\xra{cl(\l_m)\leq 0}
  (r_m,\val'')$. For the transition to be enabled, we need to have
  $\val''(cl(\l_m))=0$. Let $(q_{j-1}, v_{j-1}) \xra{} (q_j, v_j)$ be
  the transition that resets either $cl(\l_m)$ or
  $cl(\bar{\l_m})$. Notice that time cannot elapse between $(q_j,
  v_j)$ and $(r_{m-1}, \val'')$. So the time elapse should have
  occured between $(r_n, \val)$ to $(q_{j-1}, v_{j-1})$. Thus it
  should be clock $cl(\l_m)$ that is reset in the transition
  $(q_{j-1}, v_{j-1}) \xra{} (q_j, v_j)$. From the above definition of
  $\chi$, we have $\l_m$ evaluating to $true$ with $\chi$ and hence
  $C_m$ evaluates to true with $\chi$ too. This holds for all the
  clauses.  This shows that $\phi$ is satisfiable with $\chi$ being
  the satisfying assignment.

\end{proof}

The \NP-hardness of $\NZP^{LU}$ then follows due to the small size of
$\SGLU(\ANZ_\phi)$.

\begin{theorem}
  \label{thm:non-Zenoness-NPC-LU-abstract-graphs}
  The abstract zone graph $\SGLU(\ANZ_\phi)$ is isomorphic to
  $\ANZ_\phi$. The non-Zenoness problem is \nph\ for abstractions
  $\ExtraLU$ and $\ExtraLUp$.
\end{theorem}

\begin{proof}
  We first prove that $\SGLU(\ANZ_\phi)$ is isomorphic to
  $\ANZ_\phi$. For every clock $x$, $L(x)=-\infty$, hence $\ExtraLU$
  abstracts all the constraints $x_i-x_j\preccurlyeq_{ij} c_{ij}$ to
  $x_i-x_j<\infty$ except those of the form $x_0-x_i\preccurlyeq_{0i}
  c_{0i}$ that are kept unchanged. Due to the guards in $\ANZ_\phi$,
  for every reachable zone $Z$ in $\SG(\ANZ_\phi)$ we have
  $x_0-x_i\leq 0$ (i.e. $x_i\geq 0$). Therefore $\ExtraLU(Z)$ is the
  zone defined by $\bigwedge_{x\in X} x\geq 0$ which is $\Rpos^X$. For
  each state of $\ANZ_\phi$, the zone $\Rpos^X$ is the only reachable
  zone in $\SGLU(\ANZ_\phi)$, hence showing the isomorphism.

  \NP-hardness then follows from
  Lemma~\ref{lem:red-3SAT-non-Zenoness}. The result transfers to
  $\ExtraLUp$ thanks to Theorem~\ref{thm:inclusion-abstractions}.
\end{proof}

Notice that the type of zero checks in $\ANZ_\phi$ is crucial to
Theorem~\ref{thm:non-Zenoness-NPC-LU-abstract-graphs}. Replacing
zero-checks of the form $x \leq 0$ by $x=0$ does not modify the
semantics of $\ANZ_\phi$. However, this yields $L(x)=0$ for every
clock $x$. Hence, the constraints of the form $x_i-x_j\leq 0$ are not
abstracted: $\ExtraLU$ then preserves the ordering among the
clocks. Each sequence of clock resets leading from $q_0$ to $q_k$
yields a distinct ordering on the clocks. Thus, there are
exponentially many LU-abstracted zones with state $q_k$. As a
consequence, the polynomial reduction from 3SAT is lost. We indeed
provide in Section~\ref{sec:gzg-lu} below an algorithm for detecting
non-Zeno runs from $\SGLU(\Aa)$ that runs in polynomial time when
$L(x) \geq 0$ for all clocks $x$. On the other hand, notice that
changing $x=0$ to $x \le 0$ reduces the size of the abstract zone
graph, in some cases, by an exponential amount. We will see in
Section~\ref{sec:discussion} how this has led to an improvement in the
reachability analysis for timed automata.

\section{Finding non-Zeno runs}
\label{sec:gzg-lu}

Recall the non-Zenoness problem ($\NZP^\abs$):

\begin{quote}
  Given an automaton $\Aa$ and its abstract zone graph
  $\SG^\abs(\Aa)$, decide if $\Aa$ has a non-Zeno run.
\end{quote}
A standard solution to this problem involves adding one auxiliary
clock to $\Aa$ to detect non-Zenoness~\cite{Tripakis:TOCL:2009}. This
solution was shown to cause an exponential blowup
in~\cite{Herbreteau:FMSD:2012}. In the same paper, a polynomial
method has been proposed in the case of the $\ExtraM$ abstraction. We
briefly recall this construction below.

An infinite run of the timed automaton could be Zeno due to two
factors:
\begin{iteMize}{$\bullet$}
\item \emph{blocking clocks}: these are clocks bounded from above
  (i.e. $x\le c$ for some $c>0$) infinitely often in the run, but are
  reset only finitely many times,
\item \emph{zero checks}: these are guards of the form $x \le 0$ or
  $x=0$ that occur infinitely often in a manner that prevents time elapse in the
  run.
\end{iteMize}
To solve $\NZP^\abs$, the task is to find if there exists an infinite
run in $\SG^\abs(\Aa)$ that neither has blocking clocks nor
zero-checks that prevent time-elapse.  The method
in~\cite{Herbreteau:FMSD:2012} tackles these two problems as
follows.  Blocking clocks are handled by first detecting a maximal
strongly connected component (SCC) of the zone graph and repeatedly
discarding the transitions that bound some blocking clock until a
non-trivial SCC with no such clocks is obtained. This algorithm runs
in time polynomial for every abstraction.  For zero checks, a
\emph{guessing zone graph} construction has been introduced to detect
nodes where time can elapse.

\subsection{Guessing zone graph \texorpdfstring{$\GZG^\abs(\Aa)$}{GZGabs(Aa)}}
\label{sec:rgzg}

The necessary and sufficient condition for time elapse in a node
inspite of zero-checks is to have every reachable zero-check from that
node preceded by a corresponding reset
(cf. Figure~\ref{fig:gzg-intuition}).

\begin{figure}[!h]
  \begin{tikzpicture}[line width=1pt,node distance=1mm]
    \begin{scope}
      \tikzstyle{every node}=[circle,fill=white,inner sep=0mm]
      \node (s0) at (0,0) {$\bullet$};         
      \node (tick) [above=of s0] {$\surd$};
      \node (s1) at (1,0) {$\bullet$};
      \node (s2) at (2,0) {$\bullet$};
      \node (s3) at (3,0) {$\bullet$};
      \node (s4) at (4,0) {$\bullet$};
    \end{scope}
    \begin{scope}
      \draw[->] (-0.5,0) -- (s0);
      \draw[dashed] (s0) -- (s1);
      \draw[->] (s1) -- node[above] {\tiny $\{x\}$} (s2);
      \draw[dashed] (s2) -- (s3);
      \draw[->] (s3) -- node[above] {\tiny $x=0$} (s4);
      \draw[dashed] (s4) -- (5,0);
    \end{scope}
  \end{tikzpicture}
  \caption{Time can elapse in the node $\surd$}
  \label{fig:gzg-intuition}
\end{figure}

Therefore, the aim is to check if there exists a node $(q,Z)$ in
$\SG^\abs(\Aa)$ such that there is a path from $(q,Z)$ back to itself
in which every zero-check is preceded by a corresponding reset. This
would instantiate to an infinite run of $\Aa$ that can elapse time
despite the zero-checks.

This is what the guessing zone graph construction achieves. The nodes
of the guessing zone graph are triples $(q,Z,Y)$ where $Y \incl X$ is
a set of clocks. The sets $Y$ are called the \emph{guess
  sets}. Whenever a clock is reset, it is added to the guess set of
the resulting node. A transition with a zero-check can be crossed only
if the clock that is checked for zero is already present in the guess
set, that is, if it was reset somewhere in the past. The guess set $Y$
in a node $(q,Z,Y)$ therefore gives the set of clocks that can
potentially be checked for zero before being reset in a path starting
from $(q,Z,Y)$. In particular, clocks that are not in $Y$ cannot be
checked for zero in the future before being reset. Hence, on a path
from a node with an empty guess set, all the zero checks are preceded
by the corresponding reset, and time can elapse in that node.

For a valuation $v$, we write $v \models (X - Y > 0)$ for the
constraint saying that all the variables in $X - Y$ are greater than
$0$ in $v$, that is: $v \sat \left( \bigwedge_{x \in (X - Y)} v(x) > 0
\right)$ . For every transition $t = (q, g,R,q')$ of $\Aa$,
$\GZG^\abs(\Aa)$ has a transition $(q,Z,Y) \tto^t_\abs (q',Z',Y')$
only if:
\begin{iteMize}{$\bullet$}
\item there is a transition $(q,Z) \tto^t_\abs (q',Z')$ in
  $\SG^\abs(\Aa)$;
\item there is a valuation $\val \in Z$ such that $v \sat (X - Y > 0)$ and $v \sat g$;
\item and $Y' = Y \cup R$.
\end{iteMize}
Observe that if the guess set is empty in a node, then the following
transition can be taken by a valuation that has all clocks greater
than zero. This shows that if there is a path from a node $(q,Z,
\es)$, zero-checks do not hinder time-elapse in this node. When a
clock is reset, this is remembered in $Y'$. This in turn allows
the clock to be checked for zero from $(q',Z',Y')$.

The guessing zone graph also contains special transitions:
\begin{iteMize}{$\bullet$}
\item $(q,Z,Y) \tto^{\tau}_{\abs} (q,Z,Y')$ with $Y' = \emptyset$ or
  $Y' = Y$.
\end{iteMize}
Hence, from any node $(q,Z,Y)$, by taking a $\tau$ transition that
leads to $(q,Z, \emptyset)$, one can non-deterministically check if
there is a path from that node where every zero-check is preceded by a
corresponding reset. 

Figure~\ref{fig:A1_GZG} depicts a timed automaton $\Aa_1$ along with
its zone graph $ZG^\abs(\Aa_1)$ and the reachable part of its guessing
zone graph $GZG^\abs(\Aa_1)$ where $\tau$-loops have been omitted. The
loop that checks $x$ for zero is disabled from node $(1, x=z,
\emptyset)$ since $x$ does not belong to the guess set. This indicates
that it is not possible to let time elapse and then take this
transition. Time can elapse in every node with an empty guess set
(nodes with $\emptyset$ as a third component) since, by construction,
every zero check must be preceded by the corresponding reset. In
particular, the cycle $(2, x-z \ge 1, \emptyset) \tto_\abs (3, x-z \ge
1, \{z\}) \tto_\abs (2, x-z \ge 1, \{z\}) \tto^\tau_\abs (2, x-z \ge
1, \emptyset)$ is the suffix of a non-Zeno run.

\input{fig/gzg-eg.tex}

It has been shown in~\cite{Herbreteau:FMSD:2012} that the number
of guess sets for every node $(q,Z)$ reachable in $\SG^\abs(\Aa)$ is
bound by $|X| + 1$ when the abstraction $\abs$ is $\ExtraM$. The case
of other abstractions was not considered. The same construction does
not give polynomial complexity even for $\ExtraMp$. We first optimize
this construction by considering an arbitrary abstraction.

\subsection{Reduced guessing zone graph \texorpdfstring{$\RGZG^\abs(\Aa)$}{RGZGabs(Aa)}}

The reduced guessing zone graph is a slight modification that
restricts the guess sets to a subset of the set of clocks.  A clock
that is never checked for zero need not be remembered in sets $Y$. We
restrict $Y$ sets to only contain clocks that can indeed be checked
for zero and we show that this is sound and complete for non-Zenoness.

We say that a clock $x$ is \emph{relevant} if there exists a guard $x
\le 0$ or $x=0$ in the automaton. We denote the set of relevant clocks
by $\Rl(\Aa)$. For a zone $Z$, let $\Cc_0(Z)$ denote the set of clocks
$x$ such that there exists a valuation $\val \in Z$ with
$\val(x)=0$. The clocks that can be checked for zero before being
reset in a path from $(q,Z)$, lie in $\Rl(\Aa)\cap \Cc_0(Z)$.

\begin{definition}[Reduced guessing zone graph]
  \label{defn:rgzg}
  Let $\Aa$ be a timed automaton with clocks $X$. The \emph{reduced
    guessing zone graph} $\RGZG^\abs(\Aa)$ has nodes of the form
  $(q,Z,Y)$ where $(q,Z)$ is a node in $\SG^\abs(\Aa)$ and $Y\incl
  \Rl(\Aa)\cap \Cc_0(Z)$. The initial node is $(q_0,Z_0,\Rl(\Aa))$,
  with $(q_0,Z_0)$ the initial node of $\SG^\abs(\Aa)$. The
  transitions are as follows:

  \medskip

 \begin{iteMize}{$\bullet$}

 \item  For $t=(q, g,
  R, q')$, there is a transition $(q,Z,Y)\tto^t_\abs (q',Z',Y')$ with:
  \begin{align*}
  Y'=(Y\cup R)\cap \Rl(\Aa)\cap \Cc_0(Z')
  \end{align*}
  if there is
  $(q,Z)\tto^t_\abs (q',Z')$ in $\SG^\abs(\Aa)$ and some valuation
  $\val\in Z$ such that $\val\sat (\Rl(\Aa)-Y)>0$ and $\val\sat g$.

\medskip

 \item  A
  new auxiliary letter $\t$ is introduced that adds transitions
  $(q,Z,Y)\tto^\t_\abs (q,Z,Y')$ for $Y'=\es$ or $Y'=Y$.
\end{iteMize}
\end{definition}

Observe that similar to the case of the guessing zone graph, 
 we require $\val\sat (\Rl(\Aa)-Y)>0$ and $\val\sat g$
for some $\val\in Z$, a transition that checks $x\le 0$ (or $x=0$) is
allowed from a node $(q,Z,Y)$ only if $x\in Y$. Thus, from a node
$(q,Z,\es)$ every reachable zero-check $x=0$ should be preceded by a
transition that resets $x$, and hence adds it to the guess set. Such a
node is called \emph{clear}. The presence of such nodes would ensure a
time-elapse even in the presence of zero-checks.

\input{fig/rgzg.tex}

Figure~\ref{fig:rgzg} shows the reachable part of the reduced guessing
zone graph $\RGZG^\abs(\Aa_1)$ for the automaton $\Aa_1$ in
Figure~\ref{fig:A1_GZG}. Notice that since $y \not\in \Rl(\Aa_1)$, the
clock $y$ is not added to the sets $Y$ when it is reset. Observe also
that clock $x$ cannot belong to the guess sets in the nodes $(3, x-z
\ge 1)$ and $(2, x-z \ge 1)$ as $x > 0$. The resulting reduced
guessing zone graph $\RGZG^\abs(\Aa_1)$ is a lot smaller than the
guessing zone graph $\GZG^\abs(\Aa_1)$ in Figure~\ref{fig:A1_GZG}. It
still contains all the information needed to detect non-Zeno runs.

Before we prove our result about the reduced guessing zone graph, we
define some notions.

\begin{definition}
 A node $(q,Z,Y)$ of $\RGZG^\abs(\Aa)$ is called \emph{clear} if the
 third component is empty: $Y = \emptyset$. A
variable $x$ is \emph{bounded} in a transition of $\RGZG^\abs$ if the
guard of the transition implies $x \le c$ for some constant $c$. A
path of $\RGZG^\abs$ is said to be \emph{blocked} if there is a
variable that is bounded infinitely often and reset only finitely
often by the transitions on the path. Otherwise the path is called
\emph{unblocked}.
\end{definition}

 An unblocked path says that there are no blocking
clocks to bound time and clear nodes suggest that inspite of
zero-checks that might possibly occur in the future, time can still
elapse. We get the following theorem.

\begin{prop}
  \label{prop:rgzg-nz}
  A timed automaton $\Aa$ has a non-Zeno run iff there exists an
  unblocked path in $\RGZG^\abs(\Aa)$ visiting a clear node infinitely
  often.
\end{prop}

The proof of Proposition~\ref{prop:rgzg-nz} is in the same lines as
for the guessing zone graph in~\cite{Herbreteau:FMSD:2012}. It follows
from the following two lemmas.

\begin{lemma}
  \label{lem:TBA-to-rgzg}
  If $\Aa$ has a non-Zeno run, then in $\RGZG^\abs(\Aa)$ there is an
  unblocked path visiting a clear node infinitely often.
\end{lemma}

\begin{proof}
  Let $\r$ be a non-Zeno run of $\mathcal{A}$:
  \begin{equation*}
    (q_0,\val_0)\xra{\d_0,t_0}
    (q_1,\val_1)\xra{\d_1,t_1}
    \cdots
  \end{equation*}
  Since $\abs$ is complete, $\r$ is an instantiation of a path $\pi$
  in $\SG^\abs(\Aa)$:
  \begin{equation*}
    (q_0,Z_0)\tto^{t_0}_\abs
    (q_1,Z_1)\tto^{t_1}_\abs
    \cdots
  \end{equation*}

  Let $\s$ be the following sequence of transitions:
  \begin{equation*}
    (q_0,Z_0,Y_0)\tto^\t_\abs
    (q_0,Z_0,Y_0')\tto^{t_0}_\abs
    (q_1,Z_1,Y_1)\tto^\t_\abs
    (q_1,Z_1,Y_1')\tto^{t_1}_\abs
    \cdots
  \end{equation*}
  where $Y_0=Rl(\Aa)$, $Y_i$ is determined by the transition relation
  in $\RGZG^\abs(\Aa)$, and $Y'_i=Y_i$ unless $\d_i>0$ when we put
  $Y'_i=\es$.

  Since $\rho$ is non-Zeno, there are infinitely many $i$ such that
  $\d_i > 0$, hence $\s$ contains infinitely many clear nodes with
  $Y'_i=\es$. From the non-Zenoness of $\rho$, we also get that $\s$
  is unblocked.

  Now, it remains to show that $\s$ is indeed a path in
  $\RGZG^\abs(\Aa)$. For this we need to see that every transition
  $(q_i,Z_i,Y_i')\tto^{t_i}_\abs (q_{i+1},Z_{i+1},Y_{i+1})$ is
  realizable from a valuation $\val\in Z_i$ such that both $\val \sat
  (Rl(\Aa)-Y'_i)>0$ and $\val \sat g_i$ where $g_i$ is the guard of
  $t_i$. We prove this by an induction on the run. As by the
  definition of $\r$, $\val_i+\d_i\sat g_i$ for all $i\geq 0$, we only
  need to prove that $\val_i+\d_i\sat (\Rl(\Aa)-Y_i')>0$. This is
  clearly true for valuation $\val_0+\d_0 \in Z_0$.

  Assume that $\val_i + \d_i \sat (\Rl(\Aa) - Y_i') > 0$. We now prove
  that $\val_{i+1} + \d_{i+1} \sat (\Rl(\Aa) - Y'_{i+1}) >
  0$. Firstly, observe that $Y_{i+1} = (Y_i' \cup R_i) \cap
  \Cc_0(Z_{i+1}) \cap \Rl(\Aa)$. Therefore a clock $x \in \Rl(\Aa) -
  Y_{i+1}$ either belongs to $\Rl(\Aa) - Y_i'$ in which case it is
  greater than $0$ by induction hypothesis, or otherwise we have $x
  \in Y_i'$ but $x \notin \Cc_0(Z_{i+1})$. By the definition of
  $\Cc_0(Z_{i+1})$, all valuations $\val \in Z_{i+1}$ satisfy $\val(x)
  > 0$ and so in particular, $\val_{i+1}(x) >0$. This leads to
  $\val_{i+1} \sat (\Rl(\Aa) - Y_{i+1}) > 0$ which easily extends to
  $\val_{i+1} + \d_{i+1} \sat (\Rl(\Aa) - Y'_{i+1}) > 0$.
\end{proof}

\begin{lemma}
  \label{lem:rgzg-to-TBA}
  Suppose $\RGZG^\abs(\Aa)$ has an unblocked path visiting 
  a clear node infinitely often then $\Aa$ has a non-Zeno run.
\end{lemma}

\begin{proof}
  The proof follows the same lines as the proof of Lemma 6 in
  \cite{Herbreteau:FMSD:2012} with the additional information that
  for all clocks $x$ that do not belong to $\Rl(\Aa)$, we have $g
  \wedge (x >0)$ consistent for every guard $g$.

  Let $\pi: (q_0, Z_0, Y_0) \tto^{t_0}_\abs \dots$ be the unblocked
  path of $\RGZG^\abs(\Aa)$ that visits a clear node infinitely
  often. Since $\abs$ is sound, take an instantiation $\rho: (q_0,
  \val_0) \xra{\d_0, t_0} \dots$ of $\Aa$. If $\rho$ is non-Zeno, we
  are done.

  Assume that $\rho$ is Zeno. It has a suffix where less than $1/2$
  time unit elapses. Let $X^r$ denote the set of clocks that are reset
  infinitely often on $\rho$. We can thus find an index $m$ such that
  $\val_n(x) < 1/2$ for all $x \in X^r$ and for all $n \ge m$. Take
  indices $i, j \ge m$ such that $Y_i = Y_j = \emptyset$ and all
  clocks in $X^r$ are reset between $i$ and $j$. We look at the
  sequence $(q_i, \val_i) \xra{\d_i, t_i} \dots (q_j,\val_j)$ and
  claim that every sequence of the form
  \begin{equation*}
    (q_i,\val'_i)\xra{\d_i,t_i}
    (q_{i+1},\val'_{i+1})\xra{\d_{i+1},t_{i+1}} \dots(q_j,\val'_j)
  \end{equation*}
  is a part of a run of $\Aa$ provided there is $\zeta\in \Rpos$ such
  that the following three conditions hold for all $k=i,\dots,j$:
  \begin{enumerate}[(1)]
  \item $\nu'_k(x)= \nu_k(x)+\zeta+1/2$ for all $x\not\in X^r$,

  \item $\nu'_k(x)=\nu_k(x)+1/2$ if $x\in X^r$ and $x$ has not been
    reset between $i$ and $k$.
  \item $\nu'_k(x)=\nu_k(x)$ otherwise, i.e., when $x\in X^r$ and $x$
    has been reset between $i$ and $k$.
  \end{enumerate}

  It is easy to see that the run obtained by replacing every such
  $i-j$ interval of $\rho$ by the above sequence gives a non-Zeno run,
  since a $1/2$ time unit has been elapsed infinitely often.

  We now show that the above is indeed a valid run of $\Aa$. For this
  we need to first show that $\val'_k + \d_k$ satisfies the guard in
  $t_k$. Let $g$ be the guard.

  For $x \not \in X^r$, from the assumption that $\rho$ is unblocked,
  we know that $g$ could only be of the form $x > c$ or $x \ge c$. So
  $\val'_k(x)$ clearly satisfies $g$. If $x \in X^r$ and is reset
  between $i$ and $k$, $\val'_k(x) = \val_k(x)$ and so we are
  done. Consider the case when $x \in X^r$ and is not reset between
  $i$ and $k$. Observe that $x \not \in Y_k$.  This is because
  $Y_i=\es$, and then only variables that are reset are added to
  $Y$. Since $x$ is not reset between $i$ and $k$, it cannot be in
  $Y_k$. By definition of transitions in $\RGZG^\abs(\Aa)$, if $x \in
  \Rl(\Aa)$ this means that $g\land (x>0)$ is consistent. But for $x
  \not \in \Rl(\Aa)$ by definition, $g \land (x >0 )$ is
  consistent. We have that $0\leq (\val_k+\d_k)(x)< 1/2$ and $1/2\leq
  (\val_k'+\d_k)(x)< 1$. So $\val_k'+\d_k$ satisfies all the
  constraints in $g$ concerning $x$ as $\val_k+\d_k$ does.

  It can also be seen that the valuation obtained from $\val'_k$ by
  resetting the clocks in transition $t_k$ is the valuation
  $\val'_{k+1}$.
\end{proof}

\subsection{Polynomial algorithms for \texorpdfstring{$\NZP^\abs$}{NZPabs}}
\label{sec:polynomial-nzp}

Since we have a node in $\RGZG^\abs(\Aa)$ for every $(q,Z)$ in
$\SG^\abs(\Aa)$ and every subset $Y$ of $\Rl(\Aa)$, it can in
principle be exponentially bigger than $\SG^\abs(\Aa)$. Below, we see
that depending on abstraction $\abs$, not all subsets $Y$ need to be
considered. Let us first define the notion of a zone ordering clocks.

\begin{definition}
Let $X'$ be a subset of $X$. We say that a zone $Z$ \emph{orders the
  clocks in $X'$} if for all clocks $x,y\in X'$, $Z$ implies that at
least one of $x\le y$ or $y\le x$ hold, that is either all valuations $v \in Z$
  satisfy $v(x) \le v(y)$ or all valuations $v \in Z$ satisfy $v(y)
  \le v(x)$. 
\end{definition}

\begin{definition}[Weakly order-preserving abstractions]
  An abstraction $\abs$ \emph{weakly preserves orders} if for all
  clocks $x,y\in \Rl(\Aa)\cap \Cc_0(Z)$, $Z\sat x\leq y$ iff
  $\abs(Z)\sat x\leq y$.
\end{definition}

It has been observed in~\cite{Herbreteau:FMSD:2012} that all the
zones that are reachable in the unabstracted zone graph $\SG(\Aa)$
order the entire set of clocks $X$. Assume that $\abs$ weakly
preserves orders, then for every reachable node $(q,Z,Y)$ in
$\RGZG^\abs(\Aa)$, the zone $Z$ orders the clocks in $\Rl(\Aa)\cap
\Cc_0(Z)$. We now show that $Y$ is downward closed with respect to
this order given by $Z$: for clocks $x, y \in \Rl(\Aa) \cap \Cc_0(Z)$,
if $Z\sat x\leq y$ and $y\in Y$, then $x\in Y$. This entails that
there are at most $|\Rl(\Aa)|+1$ downward closed sets to consider,
thus giving a polynomial complexity.

\begin{prop}
  \label{prop:order-preserving-yield-polynomial}
  Let $\Aa$ be a timed automaton. If $\abs$ weakly preserves orders,
  then the reachable part of $\RGZG^\abs(\Aa)$ is $\Oo(|\Rl(\Aa)|)$
  bigger than the reachable part of $\SG^\abs(\Aa)$.
\end{prop}

\begin{proof}
  We prove by induction on the transitions in $\RGZG^\abs(\Aa)$ that
  for every reachable node $(q,Z,Y)$ the set $Y$ is downward closed
  with respect to the order on the clocks in $\Rl(\Aa)\cap \Cc_0(Z)$
  implied by $Z$. This is true for the initial node $(q_0,Z_0,\Rl(\Aa))$.

  Now, assume that this is true for $(q,Z,Y)$. Take a transition
  $(q,Z,Y)\tto^{t}_\abs (q',Z',Y')$ with $t=(q,g,R,q')$. By definition,
  $Y'=(Y\cup R)\cap \Rl(\Aa)\cap \Cc_0(Z')$. Suppose $Z'\sat x\le y$
  for some $x,y\in \Rl(\Aa)\cap \Cc_0(Z')$ and suppose $y\in Y'$. This
  could mean $y\in Y$ or $y\in R$. If $y\in R$, then $x$ is also in
  $R$ since $Z'\sat x\leq y$. If $y\notin R$ then we get $y\in Y$ and
  $Z\sat x\leq y$. By hypothesis that $Y$ is downward closed, $x\in
  Y$. In both cases $x\in Y'$.
\end{proof}

The following lemma shows that the $M$-extrapolations weakly
preserve orders. Hence, $\RGZGM(\Aa)$ yields a polynomial algorithm
for $\NZP^M$. Thanks to the reduction of the guessing zone graph to
the relevant clocks, this algorithm is more efficient than the
algorithm in~\cite{Herbreteau:FMSD:2012} even while using the same
abstraction.

\begin{theorem}
  \label{thm:ExtraM-order-preserving}
  The abstractions $\ExtraM$ and $\ExtraMp$ weakly preserve
  orders. The non-Zenoness problem is solved in polynomial time for
  $\ExtraM$ and $\ExtraMp$.
\end{theorem}

\begin{proof}
  It has been proved in \cite{Herbreteau:FMSD:2012} that $\ExtraM$
  weakly preserves orders. We now prove this for $\ExtraMp$. Firstly
  note that for a clock $x$ in $\Rl(\Aa)$ we have $M(x) \ge
  0$. Moreover if $x\in \Cc_0(Z)$ we have that $Z$ is consistent with
  $x\le 0$. Hence, for a clock $x \in \Rl(\Aa) \cap \Cc_0(Z)$, $Z$ is
  consistent with $x \le M(x)$.  Therefore, by definition,
  $\ExtraMp(Z)$ restricted to clocks in $\Rl(\Aa) \cap \Cc_0(Z)$ is
  identical to $\ExtraM(Z)$ restricted to the same set of
  clocks. Since $\ExtraM$ weakly preserves orders, we get that
  $\ExtraMp$ weakly preserves orders too.

  The algorithm in Proposition~\ref{prop:rgzg-nz} is thus polynomial
  for $\ExtraM$ and $\ExtraMp$ by
  Proposition~\ref{prop:order-preserving-yield-polynomial}.
\end{proof}

However, the polynomial complexity is not preserved by the coarser
$LU$-extrapolations.

\begin{thm}
  \label{thm:nzp-np-extralu}
  The abstractions $\ExtraLU$ and $\ExtraLUp$ do not weakly preserve
  orders. The non-Zenoness problem is \npc{} for $\ExtraLU$ and
  $\ExtraLUp$.
\end{thm}

\begin{proof} 
  From the definition of $\ExtraLU$ we have that all constraints of
  the form $x_i - x_j \fleq_{ij} c_{ij}$ are abstracted to $x_i - x_j
  < \infty$ when $L(x_i) = - \infty$. Thus, information about the
  relative ordering between $x_i$ and $x_j$ is lost. This shows that
  $\ExtraLU$ does not weakly preserve orders when $L(x) = - \infty$
  for relevant clocks. This also holds for $\ExtraLUp$ by
  Theorem~\ref{thm:inclusion-abstractions}.

  \NP\ hardness is proven in
  Theorem~\ref{thm:non-Zenoness-NPC-LU-abstract-graphs}. It remains to
  discuss \NP\ membership. Let $N$ be the number of nodes in
  $\SGLU(\Aa)$. Let us non-deterministically choose a node $(q,Z)$. We
  assume that $(q,Z)$ is reachable as this can be checked in
  polynomial time on $\SGLU(\Aa)$.

  We augment $(q,Z)$ with an empty guess set of clocks. From the node
  $(q,Z,\es)$, we non-deterministically simulate a path $\pi$ of the
  (non-reduced) guessing zone graph~\cite{Herbreteau:FMSD:2012}
  obtained from Definition~\ref{defn:rgzg} with $\Rl(\Aa)=X$ and
  $\Cc_0(Z)=X$ for every zone $Z$. We avoid taking $\t$ transitions on
  this path.  This ensures that the guess sets accumulate all the
  resets on $\pi$. During the simulation, we also keep track of a
  separate set $U$ containing all the clocks that are bounded from
  above on a transition in $\pi$.

  We write $\tto^*_\abs$ to denote the transitive closure of
  $\tto_\abs$. If during the simulation one reaches a node $(q,Z,Y)$
  such that $U\subseteq Y$, then we have a cycle
  $(q,Z,\emptyset)\tto_\abs^* (q,Z,Y)\tto^\t_\abs (q,Z,\emptyset)$
  that is unblocked and that visits a clear node infinitely
  often. Also, since $(q,Z)$ is reachable in $\SGLU(\Aa)$, $(q,Z,X)$
  is reachable in the guessing zone graph. Then $(q,Z,\es)$ is
  reachable from $(q,Z,X)$ with a $\t$
  transition. From~\cite{Herbreteau:FMSD:2012} and from the fact
  that $\ExtraLU$ and $\ExtraLUp$ are sound and
  complete~\cite{Behrmann:STTT:2006} we get a non-Zeno run of $\Aa$.

  Notice that it is sufficient to simulate $N\times (|X|+1)$
  transitions since we can avoid visiting a node $(q',Z',Y')$ twice in
  $\pi$.
\end{proof}

\subsection{Modified \texorpdfstring{$LU$}{LU}-extrapolations for polynomial \texorpdfstring{$\NZP^{\abs}$}{NZPabs}}

The $LU$-extrapolations do not weakly preserve orders in zones
due to relevant clocks with $L(x)=-\infty$ and $U(x)\geq 0$.  We show
that this is the only reason for \NP-hardness. We slightly modify
$\ExtraLU$ to get an abstraction $\elbu$ that is coarser than
$\ExtraM$, but it still weakly preserves orders.

\begin{definition}[Weak $L$ bounds]
  \label{defn:weak-L-bounds}
  Let $\Aa$ be a timed automaton. Given the bounds $L(x)$ and $U(x)$
  for every clock $x\in X$, the \emph{weak lower bound} $\bar{L}$ is
  given by: $\bar{L}(x)=0$ if $x \in \Rl(\Aa)$, $L(x)=-\infty$ and
  $U(x)\geq 0$, and $\bar{L}(x)=L(x)$ otherwise.
\end{definition}

We denote $\elbu$ the $\ExtraLU$ abstraction obtained by choosing
$\bar{L}$ instead of $L$. Notice that $\elbu$ and $\ExtraLU$ coincide
when zero-checks are written $x=0$ instead of $x\leq 0$ in the
automaton. By definition of $\ExtraLU$ and
Proposition~\ref{prop:order-preserving-yield-polynomial}, we get the
following.

\begin{theorem}
  \label{thm:elbu-order-preserving}
  The abstraction $\elbu$ weakly preserves orders. The non-Zenoness
  problem is solved in polynomial time for $\elbu$.
\end{theorem}

$\elbu$ coincides with $\ExtraLU$ for a wide class of automata. For
instance, when the automaton does not have a zero-check, $\elbu$ is
exactly $\ExtraLU$, and the existence of a non-Zeno run can be decided
in polynomial time. For some automata however, the zone graph obtained
with $\elbu$ is exponentially bigger than the zone graph obtained with
$\ExtraLU$. This is for instance the case for the automaton
$\ANZ_{\phi}$ used to prove \NP-hardness of $\NZP^{LU}$ in
Section~\ref{sec:nz-npc-lu}. Similar to $\elbu$ we can define $\elbup$
which again weakly preserves orders and yield a polynomial algorithm
to solve the non-Zenoness problem.

\section{The Zenoness problem}
\label{sec:zenoness-problem}

In this section we consider the Zenoness problem ($\ZP^\abs$):
\begin{quote}
  Given an automaton $\Aa$ and its abstract zone graph $ZG^\abs(\Aa)$,
  decide if $\Aa$ has a Zeno run.
\end{quote}
As in the case of non-Zenoness, this problem turns out to be
\NP-complete when the abstraction operator $\abs$ is $\ExtraLU$. We
subsequently give the hardness proof by providing a reduction
from 3SAT.

\subsection{Reducing 3SAT to \texorpdfstring{$\ZP^\abs$}{ZPabs} with abstraction \texorpdfstring{$\ExtraLU$}{ExtraLU}}

Let $P = \{ p_1, \dots, p_k\}$ be a set of propositional
variables. Let $\phi = C_1 \wedge \dots \wedge C_n$ be a 3CNF formula
with $n$ clauses. Each clause $C_m$, $m=1,2,\dots,n$ is a disjunction
of three literals $\l_1^m, \l_2^m$ and $\l_3^m$. We construct in
polynomial time an automaton $\AZ_\phi$ and its zone graph
$\zlu(\AZ_\phi)$ such that $\AZ_\phi$ has a Zeno run iff $\phi$ is
satisfiable, thus proving the \NP-hardness.

The automaton $\AZ_\phi$ has clocks $\{x_1, \bar{x_1}, \dots, x_k,
\bar{x_k} \}$ with $x_i$ and $\bar{x_i}$ corresponding to the literals
$p_i$ and $\neg p_i$ respectively. We denote the clock associated to a
literal $\l$ by $cl(\l)$.  The set of states of $\AZ_\phi$ is given by
$\{q_0, q_1, \dots, q_k\} \cup \{r_0, r_1, r_2, \dots, r_n\}$ with
$q_0$ being the initial state.  The transitions are as follows:

\begin{iteMize}{$\bullet$}
\item transitions $q_{i-1} \xra{\{x_i \}} q_i$ and $q_{i-1}
  \xra{\{\bar{x_i}\}} q_i$ for $i=1,2,\dots,k$,
\item a transition $q_k \xra{} r_0$ with no guards and resets,
\item for each clause $C_m$ there are three transitions $r_{m-1}
  \xra{cl(\neg \l) \ge 1} r_m$ for each literal $\l \in \{\l^m_1,
  \l^m_2, \l^m_3\}$,
\item a transition $r_n \xra{} q_0$ with no guards and resets. This
  transition creates a cycle in $\AZ_\phi$.
\end{iteMize}

\input{fig/zeno.tex}

As an example, Figure~\ref{fig:zeno-eg} shows the automaton for the
formula $(p_1 \vee \neg p_2 \vee p_3) \wedge (\neg p_1 \vee p_2 \vee
p_3)$. Observe that the transitions $r_{m-1} \xra{cl(\neg \l) \ge 1}
r_m$ check if the clock corresponding to the \emph{negation} of $\l$
is greater than $1$. That is, $cl(\neg \l)=\bar{cl(\l)}$.

Clearly, $\AZ_\phi$ can be constructed from $\phi$ in $\Oo(|\phi|)$
time. We now show that $\phi$ is satisfiable iff $\AZ_\phi$ has a Zeno
run.

\begin{lemma}\label{lem:formula-automaton}
  A 3CNF formula $\phi$ is satisfiable iff $\AZ_\phi$ has a Zeno run.
\end{lemma}

\begin{proof}
  For the left-to-right direction, suppose that $\phi$ is
  satisfiable. Then there exists a variable assignment $\chi: P
  \mapsto \{true, false\}$ that evaluates $\phi$ to true. We now build
  the Zeno run of $\AZ_\phi$ using $\chi$.

  Pick an infinite run $\rho$ of $\AZ_\phi$. Clearly, it should have
  the following sequence of states repeated infinitely often:

  \begin{equation}\label{eqn:sequence}
    q_0 \xra{} \dots q_k \xra{} r_0 \xra{} r_1 \xra{} \dots r_n
  \end{equation}

  We choose the transitions for $\rho$ that allow time elapse only by
  a finite amount. If $\chi(p_i) = true$, then we put $q_{i-1}
  \xra{\{x_i\}} q_i$ wherever $q_{i-1} \xra{} q_i$ occurs in
  $\rho$. Otherwise $\chi(p_i) = false$ and we put $q_{i-1}
  \xra{\{\bar{x_i}\}} q_i$. We now need to choose the transitions
  $r_{m-1} \xra{} r_{m}$ for $m = 1, \dots, n$.  Since $\chi$ is a
  satisfying assignment, every clause $C_m$ has a literal $\l$ that
  evaluates to true with $\chi$. We choose the corresponding
  transition $r_{m-1} \xra{cl(\neg \l) \ge 1} r_m$. Observe that if
  $\l$ evaluates to true, it implies that $cl(\l)$ was reset in one of
  the $q_i \xra{} q_{i+1}$ transitions but not $cl(\neg \l)$.

  Therefore, the above construction yields a sequence of transitions
  with the property that all clocks that are reset are never checked
  for greater than 1. This sequence can be taken by elapsing 1 time
  unit in the very first state, and then subsequently elapsing no time
  at all, thus giving a Zeno run in $\AZ_\phi$.

  \bigskip

  We now prove the right-to-left direction. Let $\rho$ be an infinite
  Zeno run of $\AZ_\phi$. An infinite run should repeat the sequence
  of states given in~(\ref{eqn:sequence}). Since $\rho$ is Zeno, it
  has a suffix $\rho^s$ such that for every clock $x$ that is reset in
  $\rho^s$, $x \ge 1$ never occurs in the transitions of
  $\rho^s$. This is because if every suffix of $\rho$ contains a clock
  that is both reset and checked for greater than 1, this would mean
  that there is a time elapse of one time unit occurring infinitely
  often, contradicting the hypothesis that $\rho$ is Zeno.

  Consider a segment $S = q_0 \xra{} \dots q_n \xra{} r_0 \xra{} r_1
  \xra{} \dots r_k$ in $\rho^s$. We construct a satisfying assignment
  $\chi: P \mapsto \{true, false\}$ for $\phi$ from $S$.

  \begin{iteMize}{$\bullet$}
  \item if $S$ contains $q_{i-1} \xra{\{x_i\}} q_i$ then set
    $\chi(p_i) = true$
  \item otherwise, it implies that $S$ contains $q_{i-1}
    \xra{\{\bar{x_i}\}} q_i$ in which case we set $\chi(p_i) = false$.
  \end{iteMize}
  This shows that for a literal $\l$, if $cl(\l)$ is reset in $S$,
  then $\chi(\l) = true$.  From the property of $\rho^s$ that no clock
  that is reset is checked in a guard, for every transition $r_{m-1}
  \xra{cl(\neg \l) \ge 1} r_m$ in $S$, it is clock $cl(\l)$ that is reset
  and hence $\chi(\l) = true$. By construction of $\AZ_\phi$, $\l$ is
  a literal in $C_m$.  Therefore, we get a literal that is true in
  every clause evaluating $\phi$ to true.
\end{proof}

It remains to show that $\zlu(\AZ_\phi)$ can also be calculated in
polynomial time from $\AZ_\phi$. We indeed note that the size of the
$\SG^{LU}(\AZ_\phi)$ is the same as that of the automaton. That will
conclude the proof that a polynomial algorithm for $\ZP^{LU}$ yields a
polynomial algorithm for the 3SAT problem.

\begin{prop}
  \label{prop:zeno-hard}
  The zone graph $\zlu(\AZ_\phi)$ is isomorphic to $\AZ_\phi$. The
  Zenoness problem is \nph{} for $\ExtraLU$ and $\ExtraLUp$.
\end{prop}

\begin{proof}
  By looking at the guards in the transitions, we get that for each
  clock $x$, $L(x) = 1$ and $U(x) = -\infty$. The initial node of the
  zone graph $\zlu(\AZ_\phi)$ is $(q_0, \ExtraLU(Z_0))$ where $Z_0$ is
  the set of valuations given by $(x_1 \ge 0) \wedge (x_1 = \bar{x_1}
  = \dots = x_k = \bar{x_k})$. By definition, since for each clock
  $x$, $U(x) = -\infty$, we have $\ExtraLU(Z_0)=\Rpos^X$, the
  non-negative half-space.

  On taking a transition with a guard $x \ge 1$ from $\Rpos^X$, we
  come to a zone $\Rpos^X \wedge x\ge 1$. However, since $U(x) =
  -\infty$, $\ExtraLU(\Rpos^X \wedge x\ge 1)$ gives back
  $\Rpos^X$. Same for transitions that reset a clock. It follows that
  $\zlu(\AZ_\phi)$ is isomorphic to $\AZ_\phi$. This extends to
  $\ExtraLUp$ by Theorem~\ref{thm:inclusion-abstractions}. Then \nph{}
  immediately follows from
  Lemma~\ref{lem:formula-automaton}.
\end{proof}

In the next section, we provide an algorithm for the zenoness problem
$\ZP^\abs$ and give conditions on abstraction $\abs$ for the solution
to be polynomial.

\subsection{Finding Zeno paths}
\label{sec:zeno-P-algorithm}

We say that a transition is \emph{lifting} if it has a guard that
implies $x \ge 1$ for some clock $x$. The idea is to find if there
exists a run of an automaton $\Aa$ in which every clock $x$ that is
reset infinitely often is lifted only finitely many times, ensuring
that the run is Zeno. This amounts to checking if there exists a cycle
in $\SG(\Aa)$ where every clock that is reset is not lifted. Observe
that when $(q,Z) \overset{x \ge c}{\Longrightarrow} (q',Z')$ is a
transition of $\SG(\Aa)$, then $Z'$ remembers that $x$ has been lifted
to a value bigger than $c$, that is to say $Z'$ entails  $x \ge c$. Therefore,
if a node $(q,Z)$ is part of a cycle of our required form, then in
particular, all the clocks that are greater than 1 in $Z$ should not
be reset in the cycle.

Based on the above intuition, our solution begins with computing the
zone graph on-the-fly. At some node $(q,Z)$ the algorithm
non-deterministically guesses that this node is part of a cycle that
yields a zeno run. This node transits to what we call the \emph{slow
  mode}. In this mode, a reset of $x$ in a transition is allowed from
$(q',Z')$ only if $Z'$ is consistent with $x < 1$: there is at least
one valuation $v \in Z'$ that has $v(x) < 1$.

Before we define our construction formally, recall that we would be
working with the abstract zone graph $ZG^\abs(\Aa)$ and not
$ZG(\Aa)$. Therefore for our solution to work, the abstraction
operator $\abs$ should remember the fact that a clock has a value
greater than 1. For an automaton $\Aa$ over the set of clocks $X$, let
$\Lf(\Aa)$ denote the set of clocks that appear in a lifting
transition of $\Aa$.

\begin{definition}[Lift-safe abstractions]
  An abstraction $\abs$ is called \emph{lift-safe} if for every zone
  $Z$ and for every clock $x \in \Lf(\Aa)$, $Z\sat x \ge 1$ iff
  $\abs(Z)\sat x \ge 1$.
\end{definition}

We are now in a position to define our \emph{slow zone graph}
construction to decide if an automaton has a Zeno run.

\begin{definition}[Slow zone graph]
  Let $\Aa$ be a timed automaton over the set of clocks $X$. Let
  $\abs$ be a lift-safe abstraction. The \emph{slow zone graph} $\zzg$
  has nodes of the form $(q,Z, l)$ where $l = \{ \free, \slow \}$. The
  initial node is $(q_0,Z_0, \free)$ where $(q_0, Z_0)$ is the initial
  node of $ZG^\abs(\Aa)$. For every transition $(q,Z) \tto^{t}_\abs
  (q',Z')$ in $ZG^\abs(\Aa)$ with $t = (q, g, R, q')$, we have the
  following transitions in $\zzg$:

  \begin{iteMize}{$\bullet$}
  \item a transition $(q, Z, \free) \tto^{t}_\abs (q',Z', \free)$,
  \item a transition $(q,Z,\slow) \tto^{t}_\abs (q',Z', \slow)$ if for
    all clocks $x \in R$, $Z\wedge g$ is consistent with $x < 1$,
  \end{iteMize}
  A new letter $\tau$ is introduced that adds transitions $(q,Z,\free)
  \tto^{\t}_\abs (q,Z, \slow)$.
\end{definition}

A node of the form $(q,Z, \slow)$ is said to be a \emph{slow} node.  A
path of $\zzg$ is said to be \emph{slow} if it has a suffix consisting
entirely of slow nodes.  The $\tau$-transitions take a node $(q,Z)$
from the \emph{free} mode to the \emph{slow} mode. Note that the
transitions of the slow mode are constrained
further. Figure~\ref{fig:szg} shows an example of an automaton and
corresponding slow zone graph. The free mode is identical to the zone
graph of the automaton. However, in the slow mode, the transition $q_0
\xra{\{x\}} q_1$ is not allowed from node $(q_0, x \ge 1, \slow)$ since
$x$ has been lifted. Hence, the only infinite paths in the slow mode
instantiate the loop on state $q_0$ which correpond to the zeno runs
of the automaton. The $\tau$ transitions allow to non
deterministically guess a node which has a slow path.

\begin{figure}
  \begin{minipage}{.2\textwidth}
    \centering
    \input{fig/ta-lift-safe-Z.tex}
  \end{minipage}
  \hfill
  \begin{minipage}{.75\textwidth}
    \centering
    \input{fig/szg-ta-lift-safe-Z.tex}
  \end{minipage}
  \caption{A timed automaton (left) and corresponding slow zone graph
    (right).}
  \label{fig:szg}
\end{figure}

The correctness follows from the fact that there is a cycle in $\zzg$
consisting entirely of slow nodes iff $\Aa$ has a Zeno run. This is
detailed in the following two lemmas.

\begin{lemma}\label{lem:ta-szg}
  If $\Aa$ has a Zeno run, then there exists an infinite slow path in
  $\zzg$.
\end{lemma}

\proof
  Let $\rho$ be a Zeno run of $\Aa$:

  \begin{equation*}
    (q_0, \val_0) \xra{\d_0, t_0} (q_1, \val_1) \xra{\d_1, t_1} \dots
  \end{equation*}\smallskip
  
  \noindent Let $\pi$ be the corresponding path in $ZG^\abs(\Aa)$:
  
  \begin{equation*}
    (q_0, Z_0) \tto^{t_0}_\abs
    (q_1, Z_1) \tto^{t_1}_\abs \dots
  \end{equation*}\smallskip

  \noindent We construct an infinite slow path in $\zzg$ from the path
  $\pi$. Let $X^l$ be the set of clocks that are lifted infinitely
  often in $\pi$ and let $X^r$ be the set of clocks that are reset
  infinitely often in $\pi$. Let $\pi^i$ denote the suffix of $\pi$
  starting from the position $i$.

  Clearly, there exists an index $m$ such that all the clocks that are
  lifted in $\pi^m$ belong to $X^l$ and the ones that are reset in
  $\pi^m$ belong to $X^r$.  Since $\rho$ is Zeno, we have $X^l \cap
  X^r = \emptyset$. This shows that all the clocks that are reset in
  $\pi^m$ are never lifted in its transitions. Therefore, there exists
  an index $k \ge m$ such that for all $j \ge k$, $Z_j$ is consistent
  with $x < 1$ for all clocks $x \in X^r$ and we get the following
  path of $\zzg$:
  \begin{equation*}
    (q_0, Z_0, \free) \tto^{t_0}_\abs
    \dots
    (q_j, Z_j, \free) \tto^{\t}_\abs
    (q_j, Z_j, \slow) \tto^{t_j}_\abs
    (q_{j+1}, Z_{j+1}, \slow) \tto^{t_{j+1}}_\abs
    \dots\eqno{\qEd}
  \end{equation*}

\begin{lemma}\label{lem:szg-ta}
  If $\zzg$ has an infinite slow path, then $\Aa$ has a Zeno run.
\end{lemma}

\begin{proof}
  Let $\pi$ be the slow path of $\zzg$:
  \begin{equation*}
    (q_0, Z_0, \free) \tto^{t_1}_\abs 
    \dots 
    (q_j, Z_j, \free) \tto^{\tau}_\abs
    (q_j, Z_j, \slow) \tto^{t_j}_\abs
    (q_{j+1}, Z_{j+1}, \slow) \tto^{t_{j+1}}_\abs
    \dots
  \end{equation*}
  Take the corresponding path in $ZG^\abs(\Aa)$ and an instance $\rho
  = (q_0, \val_0) \xra{\d_0, t_0} (q_1, \val_1) \dots$ which is a run
  of $\Aa$, as we have assumed that $\abs$ is a sound abstraction.
 
  Let $X^r$ be the set of clocks that are reset infinitely often and
  let $X^l$ be the set of clocks that are lifted infinitely often in
  $\rho$. By the semantics of the slow mode and from our hypothesis of
  $\abs$ being lift-safe, after the index $j$, all clocks that are
  lifted once can never be reset again. Therefore, there exists an
  index $k \ge j$ such that the following hold:

  \begin{iteMize}{$\bullet$}
  \item all clocks that are reset in $\rho^k$ belong to $X^r$ and all
    clocks that are lifted in a transition of $\rho^k$ belong to
    $X^l$,
  \item for all $x \in X^l$ and for all $i \ge k$, $\val_i(x) \ge c$
    where $c$ is the maximum constant appearing in a lifting
    transition of $\rho^k$.
  \end{iteMize}
 
  \noindent We now modify the time delays of $\rho^k$ to construct a run that
  elapses a bounded amount of time.  Pick the sequence of indices
  $i_1, i_2, \dots$ in $\rho^k$ such that $\d_{i_m} > 0$, for all $m
  \in \Nat$. Define the new delays $\d'_i$ for all $i \ge k$ as
  follows:
  \begin{equation*}
    \d'_i = 
    \begin{cases}
      min(\d_i, \frac{1}{2^{j}}) & \text{if}~i= i_j~\text{for some }j \\
      0 & \text{otherwise}
    \end{cases}
  \end{equation*}
  Consider the run $\rho'$ obtained by elapsing $\d_i'$ time units
  after the index $k$:
  \begin{equation*}
    (q_0, \val_0) \xra{\d_0, t_0} 
    \dots 
    \xra{\d_{k-1}, t_{k-1}}
    (q_k, \val_k) \xra{\d_k', t_k}
    (q_{k+1}, \val'_{k+1}) \xra{\d'_{k+1}, t_{k+1}}
    \dots
  \end{equation*}
  Clearly, $\rho'$ is Zeno. It remains to prove that $\rho'$ is a run
  of $\Aa$.  Denote $\val_k$ by $\val'_k$. We need to show that for
  all $i \ge k$, $\val'_i + \d'_i$ satisfies the guard in the
  transition $t_i$. Call this guard $g_i$. Clearly, since $\val'_i +
  \d'_i \le \val_i + \d_i$ by definition, if $g_i$ is of form $x < c$
  or $x \le c$ then it is satisfied by the new valuation. Let us now
  consider the case when $g_i$ is of the form $x \ge c$ or $x > c$. If
  $c \ge 1$, then we know that $x \in X^l$ from the assumption on
  $k$. But since $\val_k(x) \ge c$ and $x$ is not reset anywhere in
  $\rho^k$, $\val'_i(x) \ge c$ for all $i$ and hence the new valuation
  satisfies $g_i$. We are left with the case when $g_i$ is of the form
  $x > 0$. However this follows since by definition of the new
  $\d'_i$, $\val_i' + \d_i' = 0$ iff $\val_i + \d_i = 0$.
\end{proof}

From the definition of $\zzg$ it follows clearly that for each node
$(q,Z)$ of the zone graph there are two nodes in $\zzg$: $(q,Z,
\free)$ and $(q,Z, \slow)$. We thus get the following theorem.

\begin{prop}
  \label{prop:zeno-lift-safe}
  Let $\abs$ be a lift-safe abstraction. The automaton $\Aa$ has a
  Zeno run iff $\zzg$ has an infinite slow path. The number of
  reachable nodes of $\zzg$ is atmost twice the number of reachable
  nodes in $ZG^\abs(\Aa)$.
\end{prop}

We now turn our attention towards some of the abstractions existing in
the literature. We observe that both $\ExtraM$ and $\ExtraMp$ are
lift-safe and hence the Zenoness problem can be solved using the slow
zone graph construction. However, in accordance to the \NP-hardness of
the problem for $\ExtraLU$, we get that $\ExtraLU$ is not lift-safe.

\begin{thm}
  The abstractions $\ExtraM$ and $\ExtraMp$ are lift-safe. The
  Zenoness problem is solved in polynomial time for $\ExtraM$ and
  $\ExtraMp$.
\end{thm}

\begin{proof}
  Observe that for every clock that is lifted, the bound $M$ is at
  least $1$. It is now direct from the definitions that $\ExtraM$ and
  $\ExtraMp$ are lift-safe. A polynomial algorithm is easily obtained
  from Proposition~\ref{prop:zeno-lift-safe}.
\end{proof}

\begin{thm}\label{thm:np-zeno}
  The abstractions $\ExtraLU$ and $\ExtraLUp$ are not lift-safe. The
  Zenoness problem for $\ExtraLU$ and $\ExtraLUp$ is \npc{}.
\end{thm}

\begin{proof}
  That $\ExtraLU$ and $\ExtraLUp$ are not lift-safe follows from the
  proof of Proposition~\ref{prop:zeno-hard}. We show the
  \NP-membership using a technique similar to the slow zone graph
  construction. Since $\ExtraLU$ is not lift-safe, the reachable zones
  in $\zlu(\Aa)$ do not maintain the information about the clocks that
  have been lifted. Therefore, at some reachable zone $(q,Z)$ we
  non-deterministically guess the set of clocks $W$ that are allowed
  to be lifted in the future and go to a node $(q,Z,W)$.  From now on,
  there are transitions $(q,Z,W) \tto^{t}_\abs (q',Z',W)$ when:
  \begin{iteMize}{$\bullet$}
  \item $(q,Z) \tto^{t}_\abs (q',Z')$ is a transition in $\zlu(\Aa)$,
  \item if $t$ contains a guard $x \ge c$ with $c \ge 1$, then $x \in
    W$,
  \item if $t$ resets a clock $x$, then $x \notin W$
  \end{iteMize}
  If a cycle is obtained that contains $(q,Z,W)$, then the clocks that
  are reset and lifted in this cycle are disjoint and hence $\Aa$ has
  a Zeno run.
 
  This shows that if $\Aa$ has a Zeno run we can non-deterministically
  choose a path of the above form and the length of this path is
  bounded by twice the number of zones in $\zlu(\Aa)$ (which is our
  other input). This proves the \NP-membership. The \NP-hardness is
  proven in Proposition~\ref{prop:zeno-hard}.
\end{proof}

\subsection{Weakening the U bounds}

We saw in Theorem~\ref{thm:np-zeno} that the extrapolation $\ExtraLU$
is not lift-safe. This is due to clocks $x$ that are lifted but have
$U(x) = -\infty$. These are exactly the clocks $x$ with $L(x) \ge 1$
and $U(x) = -\infty$. We propose to weaken the $U$ bounds so that the
information about a clock being lifted is remembered in the abstracted
zone.

\begin{definition}[Weak $U$ bounds]\label{def:weakening-u-bounds}
  Given the bounds $L(x)$ and $U(x)$ for each clock $x \in X$, the
  \emph{weak upper bound} $\U(x)$ is given by: $\U(x)=1$ if $L(x) \ge
  1$ and $U(x) = -\infty$, and $\U(x)=U(x)$ otherwise.
\end{definition}

Let $\elub$ denote the $\ExtraLU$ abstraction, but with $\U$ bound for
each clock instead of $U$. This definition ensures that for all lifted
clocks, that is, for all $x \in \Lf(\Aa)$, if a zone entails that $x \ge
1$ then $\elub(Z)$ also entails that $x \ge 1$. This is summarized by
the following proposition, the proof of which follows by definitions
and Proposition~\ref{prop:zeno-lift-safe}.
 
\begin{theorem}
\label{thm:elub-lift-safe}
  For all zones $Z$, $\elub$ is lift-safe. The Zenoness problem is
  solved in polynomial time for $\elub$.
\end{theorem}

The Zenoness problem is polynomial for $\elub$, however, there is a
price to pay. Weakening the $U$ bounds leads to zone graphs
exponentially bigger in some cases. For example, for the automaton
$\AZ_\phi$ that was used to prove the \NP-hardness of the Zenoness
problem with $\ExtraLU$, note that the zone graph
$\SG^{L\U}(\AZ_\phi)$ obtained by applying $\elub$ is exponentially
bigger than $\zlu(\AZ_\phi)$. This leads to a slow zone graph $\Ss
\SG^{L\U}(\AZ_\phi)$ with size polynomial in
$\SG^{L\bar{U}}(\AZ_\phi)$. Similar to $\elub$ we can define $\elubp$
which is again lift-safe and yields a polynomial algorithm for the
Zenoness problem.

\section{Discussion on \texorpdfstring{$LU$}{LU}-extrapolations}
\label{sec:discussion}

In this section, we discuss two observations arising out of the 
analysis of the non-Zeno/Zeno runs in an automaton. The first
observation relates to an optimization in the reachability and
liveness algorithms
for timed automata. For the second observation, we consider the
weak $LU$-extrapolations $\elub$ and $\elbu$ and look at when these
abstractions coincide with the $LU$-extrapolation $\ExtraLU$. We note
that this happens for a wide class of timed automata.

\subsection{Optimization}

Although this paper focuses on the complexity of finding Zeno and
non-Zeno behaviours from abstract zone graphs, our analysis showing
the \NP-hardness of the non-Zenoness problem on $LU$-abstract zone
graphs leads to an interesting side-effect for the classical
reachability and liveness problems for timed automata. We have pointed
out in the introduction that the reachability and liveness problems
are solved via the abstract zone graph. The $LU$-extrapolations are
the standard abstractions used in state-of-the-art
implementations~\cite{Behrmann:STTT:2006} as they give rise to small
abstract zone graphs. The following observation
helps in reducing the abstract zone graph even further in some cases.

Recall the proof of \NP-completeness of $\NZP^{LU}$ given in
Theorem~\ref{thm:non-Zenoness-NPC-LU-abstract-graphs}. For a 3CNF
formula $\phi$ we built an automaton $\Aa_{\phi}$ that has a non-Zeno
run iff $\phi$ is satisfiable. The rest of the proof relies on the
crucial fact that the zone graph $\SG^{LU}(\Aa_{\phi})$ is isomorphic
to $\Aa$. This was indeed possible as $L(x)$ was $-\infty$ for all $x$
thanks to the guards of the form $x \le 0$. Note that modifying $x \le
0$ to $x = 0$ does not change the semantics of the automaton, but
obliges $L(x)$ to be $0$ for all clocks. In this case, the zone graph
$\SG^{LU}(\Aa_\phi)$ is no longer isomorphic to $\Aa_\phi$ and in fact
it is exponentially larger than $\Aa_{\phi}$.

This gives us the easy optimization for analyzing an automaton $\Aa$
for both reachability and liveness. Since both these algorithms go
through the zone graph construction, reducing the abstract zone graph,
and even trying to get a zone graph isomorphic to the automaton, can
produce considerable gain. The optimization consists in changing all
the guards in $\Aa$ that are of the form $x = 0$ to $x \le 0$ and in
removing all the guards $x \ge 0$. Thus, we make sure that $L(x) = -
\infty$ unless there is a guard $x \ge c$ or $x > c$ in the
automaton. This modification has been incorporated in
UPPAAL~4.1.5. Experimental results have shown a remarkable gain, in
particular for timed Petri nets as the translation to timed automata
may generate many guards like $x=0$ and $x \ge 0$. For instance,
checking reachability on a model of Fischer's protocol only explored
$2541$ nodes instead of $23042$ nodes thanks to this optimization.

\subsection{How weak are the weak \texorpdfstring{$LU$}{LU}-extrapolations}

We saw that slightly weakening the $LU$-bounds makes the $\NZP^\abs$
and $\ZP^\abs$ polynomial (Definitions~\ref{defn:weak-L-bounds},
\ref{def:weakening-u-bounds} and Theorems
\ref{thm:elbu-order-preserving}, \ref{thm:elub-lift-safe}).
 Of course, this would indeed increase the
size of the zone graph instead. However in many cases the weak
abstractions coincide with $\ExtraLU$ and $\ExtraLUp$
respectively. This shows that in all these cases, one can use the zone
graph crafted by $\ExtraLU$ itself, or alternatively $\ExtraLUp$, and
additionally checking for Zeno behaviours is not costly either. Hence
one can expect an efficient procedure for checking Zeno behaviours for
these classes of automata.

Assume that we are given an automaton $\Aa$.  Recall the abstraction
$\elbu$. This abstraction makes use of weak $L$ bounds for every clock
(cf. Definition \ref{defn:weak-L-bounds}). If all clocks that are
checked for $x \le 0$ have a guard of the form $x \ge c$ then the weak
abstraction coincides with $\ExtraLU$. Notice that this is
particularly the case for Timed Automata that do not have zero checks
(i.e. guards like $x \le 0$). Most models of ``real systems'' do not
have such guards.

For the abstraction $\elub$, the abstraction makes use of weak $U$
bounds (cf. Definition~\ref{def:weakening-u-bounds}). Notice that if
all clocks that are checked for a lower bound guard are also checked
for an upper bound then the two abstractions coincide. So, the wide
class of systems where each clock is both bounded from above (i.e. $x
\le c$) and from below (i.e. $x \ge c'$) have polynomial-time
detection of Zeno runs, even using $\ExtraLU$ and $\ExtraLUp$.

\section{\PSPACE-Completeness of (non-)Zeno Run Detection with Input
  \texorpdfstring{$\Aa$}{Aa}}
\label{sec:pspace-complete-zeno}

In the previous sections, we have characterized the complexity of
finding (non-)Zeno runs given an automaton $\Aa$ and an abstract zone
graph $\SG^\abs(\Aa)$. We now show that in the classical setting,
where automaton $\Aa$ is the only input, the two problems turn out to
be harder. We prove the following theorem.

\begin{theorem}
  Given an automaton $\Aa$, deciding whether there exists a non-Zeno
  run is \PSPACE-complete. Similarly for deciding if there exists a
  Zeno run.
\end{theorem}

Our proof follows the same lines as the proof of \PSPACE-completeness
of the emptiness problem for timed
automata~\cite{Alur:TCS:1994,Courcoubetis:FMSD:1992}.

\section*{\PSPACE-membership}

In Theorems~\ref{thm:nzp-np-extralu} and~\ref{thm:np-zeno} we have
proved that given $\Aa$ and $\SG^{LU}(\Aa)$, there is a
non-deterministic polynomial algorithm for $\NZP^{LU}$ and $\ZP^{LU}$.
Essentially both the algorithms do the following. They begin by
non-deterministically guessing a node $(q, Z)$ of $\SG^{LU}(\Aa)$ and
augmenting it with a guessed subset of clocks $S \incl X$ to give the
node $(q,Z,S)$. Starting from this node, the algorithms construct a
cycle of $\SG^{LU}(\Aa)$ containing $(q,Z)$ and satisfying certain
constraints specified by this newly augmented component:
\begin{align*}
(q,Z,S) \tto^{t_1} (q_1, Z_1, S_1) \tto^{t_2} \dots \tto^{t_n} (q_n,
Z_n, S_n) \tto^{t} (q, Z, S')
\end{align*}
Since $Z$ can be represented in space $\Oo(|X|^2)$ using a DBM, nodes
of the form $(q,Z,S)$ can be represented in space polynomial in the
size of $\Aa$. Note that each integer in the (LU-abstracted) zones $Z$
that we consider is less than the maximum constant occuring in $\Aa$.
To find the above cycle, it is enough to maintain the initially
guessed node $(q,Z,S)$ and the current node whose successor has to be
computed. Clearly, the non-deterministic algorithm needs space that is
polynomial in the size of the input $\Aa$. By Savitch's theorem, this
shows that deciding if a timed automaton has a non-Zeno run, or dually
a Zeno run, is in \PSPACE.

\section*{\PSPACE-hardness}

The problem of deciding if a deterministic Linear Bounded
Automaton\footnote{Linear Bounded Automata are Turing Machines with
  tape bounded by the length of the input word.} (LBA) $\Bb$ accepts a
word $w$ is known to be \PSPACE-complete~\cite{Hopcroft:Book:1979}. We
reduce the acceptance problem for deterministic LBAs to the problem of
deciding if a timed automaton has a Zeno or a non-Zeno run.

Let $\Bb$ be a deterministic LBA and let $w$ be a finite word on the
input alphabet of $\Bb$. Without loss of
generality, we can assume that $\Bb$ has a single accepting state
$q_F$ from which there are no outgoing transitions. We also assume that
the tape alphabet of $\Bb$ is $\G = \{1,\dots,k-1\}$. Let $n$ be the
length of the input word $w$ (hence the size of the tape of $\Bb$).

We build a timed automaton $\Aa$ that reads the sequence $\s$ of
configurations of $\Bb$ on input $w$ encoded as:
\begin{equation*}
  \g_1^0 \g_2^0 \cdots \g_n^0 \,~ k~ \, \g_1^1 \g_2^1 \cdots \g_n^1 \,
  ~k \, \cdots \, k ~\, \g_1^{i} \g_2^{i} \cdots \g_n^{i} \,~ k \, \cdots
\end{equation*}
where:
\begin{iteMize}{$\bullet$}
\item $\g_1^0 \g_2^0 \cdots \g_n^0$ is the word $w$, which is the
  initial content in the tape;
\item $\g_1^{i} \g_2^{i} \cdots \g_n^{i}$ is the content in the tape
  after the first $i$ transitions of $\Bb$.
\end{iteMize}

\medskip

For every word $w$, there is a unique encoding $\s$ as $\Bb$ is
deterministic. Observe that $\s$ is a sequence of integers in $1,
\dots, k$ and $k$ acts as a separator between successive
configurations.  The automaton $\Aa$ that we construct below accepts
the sequence $\s$ iff $\Bb$ accepts $w$.

Call $\g_1^i \g_2^i \cdots \g_n^i$ as the $i^{th}$ block. Each block
$i$ can be mapped to a position $p_i \in \{1, \dots, n\}$ which
represents the position of the tape head after the $i^{th}$
transition. The position $p_0$ is the initial position of the tape
head which is $1$. Similarly, each block $i$ can be mapped to a state
$q_i$ of the the LBA $\Bb$ representing the state of $\Bb$ after the
first $i$ transitions. 

We construct the automaton $\Aa$ as follows.

\subsection*{States.} The states of the automaton encode the state
of $\Bb$ and the position of the tape head. So each state of the
automaton is of the form $(q, p)$ where $q$ is a state of $\Bb$ and $p
\in \{1, \dots, n \}$ is a position of the tape head. There is an
extra auxiliary state $(q_{init}, 0)$ to read the initial block of
$\s$ which is the word $w$ itself. The goal is to make the automaton
come to $(q_i, p_i)$ after reading the first $i$ blocks:

\begin{equation*}
  \underbrace{\g_1^0 \g_2^0 \cdots \g_n^0 \,~ k}_{\scriptstyle
    (q_{init}, 0)}~ 
  \underbrace{\, \g_1^1 \g_2^1 \cdots \g_n^1 \, ~k }_{\scriptstyle
    (q_0, p_0)}
  \, \cdots \, k~ 
  \underbrace{\, \g_1^{i} \g_2^{i} \cdots \g_n^{i} \,~
    k}_{\scriptstyle (q_{i-1}, p_{i-1})}~
  \underbrace{\, \g_1^{i+1} \g_2^{i+1} \cdots \g_n^{i+1} \,~
    k}_{\scriptstyle (q_{i}, p_{i})}
  \, \cdots
\end{equation*}

The initial block is read in the initial state $(q_{init}, 0)$ after
which the automaton moves to $(q_0, p_0)$. 
In general, after reading block $i$, the automaton
should move to $(q_i, p_i)$ which represents the state $q_i$ of
$\Bb$ and the position $p_i$ of the tape head at the time of taking
the $i^{th}$ transition. While reading the $i+1^{th}$ block from state
$(q_{i},p_{i})$ the automaton has to check if the symbol at position $p_{i}$
of the block corresponds to the modification of the $i+1^{th}$
transition of $\Bb$ which is of the form $(q_i, \g, \g', \D, q_{i+1})$.

\subsection*{Clocks.}
We intend to make the automaton $\Aa$ spend $k+1$ time units at each
symbol. This is facilitated by a clock $x$. Spending $k+1$ time units
will also help us to recognize the current symbol
which is a number between $1$ and $k$.
To this regard, to read a symbol $s \in \s$, we use a
transition with guard $(x=s)$ followed by a transition with guard
$(x=k+1)$ that resets $x$.
As reading a symbol requires $(k+1)$ time units, reading a tape
configuration (followed by separator symbol $k$) takes $(n+1).(k+1)$
time units.

To store the currently read symbol, we introduce a clock $x_j$ for
each cell $j$ of the tape.
If the currently read symbol is $\g_j^i$, then clock
$x_j$ is reset on the transition with guard $x = \g_j^i$.  Hence, when
the symbol $\g_j^{i+1}$ is read, the previous content of the cell $j$,
given by the 
symbol $\g_j^i$, is remembered in $x_j$ by the value
$(n+1).(k+1)-\g_j^i+\g_j^{i+1}$. This
is illustrated in~(\ref{eqn:time-elapses-successive-cell-read}).

\begin{equation}
  \label{eqn:time-elapses-successive-cell-read}
  \cdots \,
  \underbrace{
    \overbrace{\xra{(x=\g_j^i), \{x_j\}}}^{\g_j^i\ t.u.}
    \, \xra{(x=k+1), \{x\}}
    \, \cdots \,
  }_{(n+1)\cdot (k+1)\ time\ units}
  \overbrace{\xra{(x=\g_j^{i+1}), \{x_j\}}}^{\g_j^{i+1}\ t.u.}
  \, \xra{(x=k+1), \{x\}}
  \, \cdots
\end{equation}

\subsection*{Transitions.}

Consider a state $(q,p)$ of $\Aa$. For each transition $(q, \g, \g',
\D, q')$ of $\Bb$, there is a sequence of transitions in $\Aa$ that
reads a block and does the following:
\begin{iteMize}{$\bullet$}
  \item ensures that the $p^{th}$ symbol corresponds to the
    modification of the $p^{th}$ tape cell forced by this transition,
  \item ensures that all other symbols are left unchanged corresponding to all
    other cells being unchanged,
  \item moves to state $(q', p+\D)$ after
    reading the block.
\end{iteMize} 
Moreover, the cells have to be read in the right order, that
is, cell 1 should be read followed by cell 2, etc. 
Recall that $x_j$ is the clock associated with every cell. For every
$j \neq p$, we check if $x_j = (n+1).(k+1)$ and for $j = p$ we check
if $x_j = (n+1).(k+1) - \g + \g'$. This will ensure the first two
conditions above and will also ensure that the cells are read in the
correct succession. 

\input{fig/pspace-widget}

The complete widget for transition $(q, \g, \g', \D, q')$ is depicted
in Figure~\ref{fig:pspace-widget}. There is one such widget in $\Aa$
for each state $(q,p)$ such that $p+\D$ is a valid position
(i.e. $p+\D \in \{1, \dots, n\}$).

\subsection*{Initialization}

We need to read the word $w$ from state $(q_{init}, 0)$ and assign the
initial value of the clocks $x_1, \dots, x_n$ to $w_1, \dots,w_n$
where $w_j$ represents the $j^{th}$ symbol of $w$. As $w$ is given as
an input, we can easily add transitions from $(q_{init}, 0)$ to ensure
this and jump to $(q_0, 1)$.

Observe that since $\Bb$ is deterministic, $\Aa$ is also
deterministic. Furthermore, $\Aa$ is time-deterministic as all the
guards are equalities. Hence, $\Aa$ has a single run given the word
$w$.   Furthermore, if
$\Bb$ does not terminate on $w$, the corresponding run of $\Aa$ is
infinite and non-Zeno.

Recall that $q_F$ is the sole accepting state of $\Bb$ and there are
no transitions outgoing from $q_F$.
From the construction described above, one easily gets the following
theorem:

\begin{theorem}
  \label{thm:lba-simulated-ta}
  $\Aa$ reaches a state $(q_F,p)$ iff $\Bb$ reaches $q_F$ on input
  $w$. The size of $\Aa$ is polynomial in the size of $\Bb$ and $w$.
\end{theorem}

\subsubsection*{Existence of a Non-Zeno Run}

We show that an algorithm for deciding if $\Aa$ has a non-Zeno run
yields an algorithm to decide if $\Bb$ accepts $w$. This algorithms
has two phases.

In the first phase, it determines if $\Aa$ has a non-Zeno run:
\begin{iteMize}{$\bullet$}
\item if the answer is \emph{yes}, we can conclude that $\Bb$ does not
  accept $w$. Indeed, if $\Aa$ has a non-Zeno run, then it does not
  reach $(q_F,p)$ for any $p$ as the run is infinite (recall $q_F$ is
  a sink state by hypothesis), hence neither does $\Bb$ reach $q_F$;
\item if the answer is \emph{no}, we cannot conclude. We only gain
  information that $\Aa$ has no infinite run, but it may stop in a
  state $(q_F,p)$ as well as in a non-accepting state.
\end{iteMize}

In the second phase, we transform $\Aa$ into $\Aa'$ by adding a loop
on all $(q_F,p)$ with guard $(x \geq 1)$ and that resets $x$. Now, if
the run of $\Aa'$ is infinite, then it visits some
$(q_F,p)$. Furthermore, it is the only non-Zeno run in $\Aa'$ as we
know from the first phase that $\Aa$ has no infinite run. We now ask
if $\Aa'$ has a non-Zeno run:
\begin{iteMize}{$\bullet$}
\item if the answer is \emph{yes}, we can conclude that $\Bb$ accepts
  $w$;
\item if the answer is \emph{no}, the run of $\Aa'$ is finite and does
  not reach any $(q_F,p)$. We can conclude that $\Bb$ does not accept
  $w$.
\end{iteMize}

\subsubsection*{Existence of a Zeno Run}

Now, we show that an algorithm that decides if $\Aa$ has a Zeno run
yields an algorithm to decide if $\Bb$ accepts $w$. Recall that $\Aa$
is deterministic: it has a unique run and if that run is
infinite, then it is non-Zeno.

We transform $\Aa$ into $\Aa'$ by adding a loop on all states
$(q_F,p)$ with guard $(x \leq 0)$. Then we ask if $\Aa'$ has a Zeno
run:
\begin{iteMize}{$\bullet$}
\item if the answer is \emph{yes}, then some $(q_F,p)$ has to be
  reachable, hence $\Bb$ reaches $q_F$ and accepts $w$;
\item if the answer is \emph{no}, then no $(q_F,p)$ is reachable, and
  $\Bb$ does not accept $w$.
\end{iteMize}

\section{Conclusion}
\label{sec:conclusion}

We have shown a striking fact that the problem of deciding existence
of Zeno or non-Zeno behaviours from abstract zone graphs depends
heavily on the abstractions, to the extent that the problem changes
from being polynomial to becoming \NP-complete as the abstractions get
coarser. Of course, it is but natural that checking for Zeno/non-Zeno
behaviours becomes difficult when the abstraction gets coarser, as
lesser information is maintained. However, the fact this difficulty
ranges from a low polynomial to \NP-hardness is surprising.

We have proved \NP-completeness for the coarse abstractions
$\ExtraLU$ and $\ExtraLUp$. In contrast, the fundamental problems of
finding accepting runs for finitary accepting conditions
(reachability), and for B\"uchi accepting conditions, over abstract
zone graphs have a mere linear complexity, independent of the
abstraction. As a consequence of the difficulty of detecting non-Zeno
runs, the Büchi emptiness problem which consists in finding a run that
is both accepting and non-Zeno is \NP-complete for abstractions
$\ExtraLU$ and $\ExtraLUp$.

On the positive side, from our study on the conditions for an
abstraction to give a polynomial solution, we see that a small
modification of the LU-extrapolation works. We have defined two weaker
abstractions: $\elbu$ for detecting non-Zeno runs and $\elub$ for
detecting Zeno runs. The weak bounds $\bar{L}$ and $\bar{U}$ can also
be used with $\ExtraLUp$ to achieve similar results. Despite leading
to a polynomial solution for checking Zeno or non-Zeno behaviours from
abstract zone graphs, these abstractions transfer the complexity to
the input: they could lead to exponentially bigger abstract zone
graphs themselves. However, for a fairly large class of automata
described in the previous section, we see that this is not the case as
the weak abstractions coincide with $\ExtraLU$.

While working with abstract zone graphs, coarse abstractions (and
hence small abstract zone graphs) are essential to handle big models
of timed automata. These, as we have seen, work against the Zenoness
questions in the general case. Our results therefore provide a
theoretical motivation to look for cheaper substitutes to the notion
of Zenoness.

All the abstractions we have considered are convex
abstractions. However, there also exist non-convex
abstractions~\cite{Bouyer:FMSD:2004,Behrmann:STTT:2006} that are known
to be coarser that the convex ones. Since non-convex sets are
particularly difficult to manipulate, only the convex abstractions
have been considered for implementation. Recently, new algorithms have
been introduced to solve the reachability problem efficiently using
non-convex
abstractions~\cite{Herbreteau:FSTTCS:2011,Herbreteau:LICS:2012}. Future
work includes adaptation of these algorithms to the detection of
(non-)Zeno behaviors.

\bibliographystyle{plain}
\bibliography{m}

\end{document}

%% file: fig/zone-boundaries.tex
\begin{tikzpicture}[scale=0.8] 
  \draw[->, very thick, >=stealth] (0,0) -- (0,5); \draw[->, very
  thick, >=stealth] (0,0) -- (6,0);

  \draw (6.3,0) node {\scriptsize $x_1$} (7, 0); \draw (0, 5.3) node
  {\scriptsize $x_2$} (0,6.3);
  \draw (-0.2,-0.2) node {\scriptsize  $0$} (0,-0.2);

  \draw[very thin,  fill=gray, nearly transparent]
  (1,1) -- (1,2) -- (3,4) -- (5,4) -- (5,3) -- (3,1) -- cycle; 
  
  \draw (0,1) -- (4,5);
  \draw[white] (4.5,5.5) -- node[sloped, black] {\footnotesize $x_2 - x_1
  < 1$} (5.5,6.5) ;
  
\draw (1,0) -- (1,5);
  \draw[white] (0.5, -0.3) -- node[black] {\footnotesize $x_1 > 1$} (1,
  -0.3);   

  \draw (0,1) -- (7,1);
  \draw[white] (7.2, 1) -- node[black] {\footnotesize $x_2 > 1$} (8,1);  

  \draw (2,0) -- (7,5);
  \draw[white] (7.5,5.5) -- node[sloped, black] {\footnotesize $x_1 - x_2
    < 2$} (8,6);

  \draw (5,0) -- (5,5);
  \draw[white] (4.5, -0.3) -- node[black] {\footnotesize $x_1 < 5$} (5,
  -0.3); 

  \draw (0,4) -- (7,4);
  \draw[white] (7.2, 4) -- node[black] {\footnotesize $x_2 < 4$} (8,4); 

\end{tikzpicture}


%% file: fig/ta-example.tex
\begin{tikzpicture}[font=\footnotesize]
  \begin{scope}
    \tikzstyle{every node}=[circle,draw]
    \node (q0) at (0,0) {$q_0$};
    \node (q1) at (2,0) {$q_1$};
    \node (q2) at (-2,0) {$q_2$};
  \end{scope}

  \begin{scope}[->,line width=1pt]
    \draw (0,0.8) -- (q0);
    \draw (q0) edge [bend left] node[above] {$\{x_1\}$} (q1);
    \draw (q1) edge [bend left] node[below] {$x_1 \le 2$} (q0);
    \draw (q0) edge node [above]  {$x_2 > 5$} (q2);
  \end{scope}
\end{tikzpicture}


%% file: fig/zg-example.tex
\begin{tikzpicture}[font=\footnotesize]
  \begin{scope}
    \tikzstyle{every node}=[draw,fill=white]
    \node (z0) at (0,0) {$q_0, \ 0 \le x_1 = x_2$};
    \node (z2) at (-5,0) {$q_2, \  5 < x_1 = x_2$};
    \node (z1) at (4,0) {$q_1, \ 0 \le x_1 \le x_2$};
    \node (z01) at (0,-2) {$q_0, \ 0 \le x_1 \le x_2$};
    \node (z21) at (-5,-2) {$q_2, \  0 \le x_1 \le x_2 \, \land \,
      x_2>5$};
  \end{scope}

  \begin{scope}[line width=1pt]
    \draw[->] (0,0.7) -- (z0);
    \draw[->] (z0) -- node[above] {$x_2 > 5$} (z2);
    \draw[->] (z0) -- node[above] {$\{x_1\}$} (z1);
    \draw[->] (z01) -- node[above] {$x_2 > 5$} (z21);
    \draw[->] (z1) edge[bend left] node[right] {$x_1\le 2$} (z01);
    \draw[->] (z01) edge node[above] {$\{x_1\}$} (z1);
  \end{scope}
\end{tikzpicture}


%% file: fig/zg-ainf.tex
\centering{
  \begin{tikzpicture}[font=\footnotesize]
    \begin{scope}
      \tikzstyle{every node} = [circle,minimum size = 0.7cm,draw]
      \node (s0) at (0,0) {$q_0$};
      \node (s1) at (2.5,0) {$q_1$};
    \end{scope}
    \begin{scope}[->,line width=1pt]
      \draw (-1,0) -- (s0);
      \draw (s1) edge[loop=above] node[above] {$x_2=1, \{x_2\}$} (s1);
      \draw (s0) edge node[above] {$\{x_1,x_2\}$} (s1);
    \end{scope}

    \begin{scope}[xshift=7cm]
      \draw[->,very thick, >=stealth] (0,0) -- (2,0);
      \draw[->,very thick, >=stealth] (0,0) -- (0,1.5);

      \node at (-0.2,-0.2) {$0$};
      \node at (2.2, -0.2) {$x_1$};
      \node at (-0.2, 1.7) {$x_2$};
      
      \node at (2.2, 0.4) {$\mathbf{\dots}$};

      \fill (0,0) circle (2pt);
        
      \draw[thin] (0,0) -- (1.4,1.4);
      
      \fill (0.4,0) circle (2pt);
      \draw [thin] (0.4,0) -- (1.8,1.4);
      
      \fill (0.8,0) circle (2pt);
      \fill (1.2,0) circle (2pt);
      \draw[thin] (0.8,0) -- (2.2,1.4);
      \draw[thin] (1.2,0) -- (2.6,1.4);
    \end{scope}
  \end{tikzpicture}

  \vspace{0.7cm}
  
  \begin{tikzpicture}[font=\footnotesize]
    \begin{scope}
      \tikzstyle{every node}=[draw]
      \node (z0) at (0,0) {$\begin{array}{c}
          q_0, \ x_1 - x_2 = 0 \\
          \land \ x_1 \ge 0
        \end{array}$};
      \node (z1) at (4.5,0) {
        $\begin{array}{c}
          q_1, \ x_1 -x_2 = 0 \\
          \land \ x_1 \ge 0
        \end{array}$};
      \node (z2) at (9, 0) {
        $\begin{array}{c}
          q_1, \ x_1 - x_2 = 1 \\
          \land \ x_1 \ge 0
        \end{array}$};
    \end{scope}
    \begin{scope}[line width=1pt]
      \draw[->] (-2, 0) -- (z0);
      \draw[->] (z0) -- node[above] {$\{x_1,x_2\}$} (z1);
      \draw[->] (z1) -- node[above] {$x_2=1$}
      node[below] {$\{ x_2\}$} (z2);
      \draw[->] (z2) -- node[above] {$x_2=1$}
      node[below] {$\{ x_2\}$} (12, 0);
    \end{scope}
    \node[right] at (12,0) {$\dots$};
  \end{tikzpicture}
}


%% file: fig/abs-example.tex
\begin{figure}[t]
  \centering
\begin{tikzpicture}

\begin{scope}[scale=0.6]
\draw[->, very thick, >=stealth] (0,0) -- (0,10); 
\draw[->, very thick, >=stealth] (0,0) -- (10,0);

\node at (-0.4, -0.4) {\scriptsize $0$};
\node at (10, -0.4) {\scriptsize $x_1$};
\node at (-0.4, 10) {\scriptsize $x_2$};


\fill[gray!30] (8,9) -- (6,9) -- (4,7) -- (4,3) -- (5,3) -- (8,6)
-- cycle;

\draw (0,3) -- (7,10);
\draw (2,0) -- (10,8);
\draw (4,7) -- (4,3) -- (5,3);
\draw (6,9) -- (8,9) -- (8,6);

\draw[thick] (1,0) -- (1,10);
\draw[thick] (3,0) -- (3,10);
\draw[thick] (0,6) -- (10,6);
\draw[thick] (0,4) -- (10,4);
\node at (1, -0.5) {\scriptsize $\mathbf{U(x_1)}$};
\node at (3, -0.5) {\scriptsize $\mathbf{L(x_1)}$};
\node at (-0.8, 6) {\scriptsize $\mathbf{U(x_2)}$};
\node at (-0.8, 4) {\scriptsize $\mathbf{L(x_2)}$};


\draw (3,3) -- (4,3);
\fill[pattern=dots] (3,6) -- (4,7) -- (4,3) -- (3,3); 
\fill[pattern=dots] (7,10) -- (6,9) -- (10,9) -- (10,10);
\fill[pattern=dots] (8,9) -- (8,6) -- (10,8) -- (10,9);


\draw (5,3) -- (10,3);
\fill[pattern= horizontal lines] (5,3) -- (10,8) -- (10,3); 
\fill[pattern= vertical lines] (3,10) -- (3,6) -- (7,10);


\draw (1,3) -- (3,3);
\fill[pattern= north west lines] (1,3) -- (1,4) -- (3,6) -- (3,3);
\fill[pattern= grid] (1,4) -- (1,10) -- (3,10) -- (3,6);

\end{scope}

\begin{scope}[xshift=7cm]
\node[left] at (2,4) {\footnotesize $Z$ :};
\node[left] at (2,3.5) {\footnotesize $\ExtraM(Z)$ :};
\node[left] at (2,3) {\footnotesize $\ExtraMp(Z)$ :};
\node[left] at (2,2.5) {\footnotesize $\ExtraLU(Z)$ :};
\node[left] at (2,2) {\footnotesize $\ExtraLUp(Z)$ :};

\foreach \y in {3.85, 3.35, 2.85, 2.35, 1.85}
  \draw[fill = gray!20] (2, \y) rectangle (2.3, \y+0.3);

\foreach \y in {3.35, 2.85, 2.35, 1.85}
{
\draw[fill, pattern=dots] (2.8, \y) rectangle (3.1, \y+0.3);
\node at (2.55, \y+0.15) {\footnotesize $\cup$};
}

\foreach \y in {2.85, 2.35, 1.85}
{
\draw[fill, pattern=horizontal lines] (3.6, \y) rectangle (3.9,
\y+0.3);
\node at (3.35, \y+0.15) {\footnotesize $\cup$}; 
}

\foreach \y in {2.85, 1.85}
{
\draw[fill, pattern=vertical lines] (4.4, \y) rectangle (4.7, \y +
0.3);
\node at (4.15, \y+0.15) {\footnotesize $\cup$}; 
}

\foreach \y in {2.35, 1.85}
{
\draw[fill, pattern=north west lines] (5.2, \y) rectangle (5.5,
\y+0.3);
\node at (4.95, \y+0.15) {\footnotesize $\cup$};
}

\draw[fill, pattern=grid] (6,1.85) rectangle (6.3, 2.15); 
\node at (5.75, 2) {\footnotesize $\cup$};

\end{scope}

\end{tikzpicture}
\caption{An illustration of the abstraction hierarchy shown in
  Figure~\ref{fig:abs_lit}.}
\label{fig:abst-examples}
\end{figure}


%% file: fig/abszg-example.tex
\begin{tikzpicture}[font=\footnotesize]
  \begin{scope}
    \tikzstyle{every node}=[draw]
    \node (z0) at (0,0) {$q_0, \ x_1 \ge 0 \land x_2 \ge 0$};
    \node (z1) at (5,0) {$q_1, \ x_1 \ge 0 \land x_2 \ge 0$};
  \end{scope}
  \begin{scope}[line width=1pt]
    \draw[->] (-2, 0) -- (z0);
    \draw[->] (z0) -- node[above] {$\{x_1, x_2\}$} (z1);
    \draw[->] (z1) edge[loop=above] node[above] {$x_2=1, \{x_2\}$}
    (z1);
  \end{scope}
\end{tikzpicture}


%% file: fig/zeno-abstract-paths.tex
\begin{tikzpicture}[font=\footnotesize]
  \begin{scope}
    \tikzstyle{every node}=[circle,draw,minimum size=5mm]
    \node (s0) at (0,0) {$q_0$};
    \node (s1) at (-2,0) {$q_1$};
    \node (s2) at (2,0) {$q_2$};
  \end{scope}
  \begin{scope}[->,line width=1pt]
    \draw (0,0.7) -- (s0);
    \draw (s0) to [bend right=15] node[above] {$\{x_1\}$} (s1);
    \draw (s0) to [bend left=15] node[above] {$\{x_2\}$} (s2);
    \draw (s1) to [bend right=15] node[below] {$x_2\leq 0$} (s0);
    \draw (s2) to [bend left=15] node[below] {$x_1\leq 0$} (s0);
  \end{scope}
  \draw (0, -1) node {Automaton $\Aa_{zeno}$.} (0.5, -1);
\end{tikzpicture}

\vspace{0.5cm}

\begin{tikzpicture}[font=\footnotesize]
  \begin{scope}
    \node (z0_0) at (0,0) {$(q_0,\ 0=x_1=x_2)$};
    \node (z1_0) at (3.2,0) {$(q_1,\ 0=x_1\leq x_2)$};
    \node (z0_1) at (6.4,0) {$(q_0,\ 0=x_1=x_2)$};
    \node (z2_0) at (9.6,0) {$(q_2,\ 0=x_2\leq x_1)$};
  \end{scope}
  \begin{scope}[->,line width=1pt]
    \draw (z0_0) -- (z1_0);
    \draw (z1_0) -- (z0_1);
    \draw (z0_1) -- (z2_0);
    \draw[dashed] (z2_0) -- (12,0);
  \end{scope}
  \draw (6,-0.7) node {A path in the abstract zone graph of
    $\Aa_{zeno}$ with abstraction $\ExtraM$.};
\end{tikzpicture}

\vspace{0.7cm}

\begin{tikzpicture}[font=\footnotesize]
  \begin{scope}
    \node (z0_0) at (0,0) {$(q_0,\top)$};
    \node (z1_0) at (2.8,0) {$(q_1,\top)$};
    \node (z0_1) at (5.6,0) {$(q_0,\top)$};
    \node (z2_0) at (8.4,0) {$(q_2,\top)$};
  \end{scope}
  \begin{scope}[->,line width=1pt]
    \draw (z0_0) -- (z1_0);
    \draw (z1_0) -- (z0_1);
    \draw (z0_1) -- (z2_0);
    \draw[dashed] (z2_0) -- (10,0);
  \end{scope}
  \draw (6,-0.7) node {A path in the abstract zone graph of
    $\Aa_{zeno}$ with abstraction $\ExtraLU$.};
\end{tikzpicture}


%% file: fig/nonzeno.tex
\begin{figure}[t]
  \centering
  \begin{tikzpicture}[shorten >=1pt,node distance=2cm,on
    grid,auto,font=\footnotesize]
  
    \node[state, initial, initial text={}] (q0) {$q_0$};
    \node[state] (q1) [right=of q0] {$q_1$};
    \node[state] (q2) [right=of q1] {$q_2$};
    \node[state] (q3) [right=of q2] {$q_3$};
    \node[state] (r0) [right=of q3] {$r_0$};
    \node[state,node distance=2.5cm] (r1) [right=of r0] {$r_1$};
    \node[state,node distance=2.5cm] (r2) [right=of r1] {$r_2$};
    
    \draw[->] (q0) edge [bend left] node {$\{x_1\}$}
    (q1); 
    \draw[->] (q0) edge [bend right] node[anchor=north]
    {$\{\bar{x_1}\}$} (q1);
   
    \draw[->] (q1) edge [bend left] node {$\{x_2\}$}
    (q2);
    \draw[->] (q1) edge [bend right] node[anchor=north]
    {$\{\bar{x_2}\}$} (q2);

    \draw[->] (q2) edge [bend left] node {$\{x_3\}$}
    (q3); 
    \draw[->] (q2) edge [bend right] node[anchor=north]
    {$\{\bar{x_3}\}$} (q3);

    \draw[->] (q3) -- (r0);

    \draw[->] (r0) edge [bend left=45] node[anchor=south] {$x_1\le
      0$} (r1);
    \draw[->] (r0) edge node[anchor=south] {$\bar{x_2}\le 0$}
    (r1);
    \draw[->] (r0) edge [bend right] node[anchor=north] {$x_3\le 0$} (r1);

    \draw[->] (r1) edge [bend left=45] node[anchor=south]
    {$\bar{x_1}\le 0$} (r2);
    \draw[->] (r1) edge node[anchor=south] {$x_2\le 0$} (r2);
    \draw[->] (r1) edge [bend right] node[anchor=north] {$x_3\le 0$}
    (r2);

    \draw[->] (r2) -- ++(0cm,-1.5cm) -- ([yshift=-1.5cm] q0.center) -- (q0);
 \end{tikzpicture}
 \caption{$\ANZ_\phi$ for $\phi = (p_1 \vee \neg p_2 \vee p_3) \wedge
   (\neg p_1 \vee p_2 \vee p_3)$}
\label{fig:nonzeno-eg}
\end{figure}


%% file: fig/gzg-eg.tex
\begin{figure}[ht]
\begin{center}
    \begin{tikzpicture}[font=\footnotesize] 
      \begin{scope}
        \tikzstyle{every node}=[circle,draw]
        \node (1) at (0,0) {$1$};
        \node (2) at (2,0) {$2$};
        \node (3) at (4,0) {$3$};
      \end{scope}
      \begin{scope}[->,line width=1pt]
        \draw (-1,0) -- (1);
        \draw (1) edge[loop=above] node[above]{$x=0,\{x\}$} (1);
        \draw (1) edge node[above]{$\{y\}$} (2);
        \draw (2) edge[bend left] node[above]{$x \ge 1, \{z\}$} (3);
        \draw (3) edge[bend left] node[below]{$z=0$} (2);
      \end{scope}
    \end{tikzpicture}

    \bigskip

    \begin{tikzpicture}[font=\footnotesize] 
      \begin{scope}
        \tikzstyle{every node}=[rectangle,draw]
        \node (1xeqz) at (0,0) {$1,x=z$};
        \node (2xeqz) at (2.5,0) {$2,x=z$};
        \node (3xmzge1) at (6,0) {$3,x-z \ge 1$};
        \node (2xmzge1) at (9,0) {$2,x-z \ge 1$};
      \end{scope}
      \begin{scope}[->,line width=1pt]
        \draw (-1.5,0) -- (1xeqz);
        \draw (1xeqz) edge[loop=above]
        node[above]{$x=0,\{x\}$} (1xeqz);
        \draw (1xeqz) edge node[above] {$\{y\}$} (2xeqz);
        \draw (2xeqz) edge node[above] {$x \ge 1, \{z\}$}
        (3xmzge1);
        \draw (3xmzge1) edge[bend left] node[above]{$z=0$}
        (2xmzge1);
        \draw (2xmzge1) edge[bend left] node[below]{$x \ge 1,
          \{z\}$} (3xmzge1);
      \end{scope}
    \end{tikzpicture}

    \bigskip
    \bigskip

    \begin{tikzpicture}[font=\footnotesize] 
      \begin{scope}
        \tikzstyle{every node}=[rectangle,draw]
        \node (1_xeqz_xyz) at (0,0) {$1, x=z, \{x,y,z\}$};
        \node (1_xeqz_clear) at (12,0) {$1, x=z, \emptyset$};
        \node (2_xeqz_xyz) at (0,-2) {$2, x=z, \{x,y,z\}$};
        \node (2_xeqz_y) at (4,-2) {$2, x=z, \{y\}$};
        \node (2_xeqz_clear) at (12,-2) {$2, x=z, \emptyset$};
        \node (3_xmzge1_xyz) at (0,-4) {$3,x-z \ge 1, \{x,y,z\}$};
        \node (3_xmzge1_yz) at (4,-4) {$3,x-z \ge 1, \{y,z\}$};
        \node (3_xmzge1_z) at (8,-4) {$3,x-z \ge 1, \{z\}$};
        \node (3_xmzge1_clear) at (12,-4) {$3,x-z \ge 1, \emptyset$};
        \node (2_xmzge1_xyz) at (0,-6) {$2,x-z \ge 1, \{x,y,z\}$};
        \node (2_xmzge1_yz) at (4,-6) {$2,x-z \ge 1, \{y,z\}$};
        \node (2_xmzge1_z) at (8,-6) {$2,x-z \ge 1, \{z\}$};
        \node (2_xmzge1_clear) at (12,-6) {$2,x-z \ge 1, \emptyset$};
      \end{scope}
      \begin{scope}[->,line width=1pt]
        \draw (-1.8,0) -- (1_xeqz_xyz);
        \draw (1_xeqz_xyz) edge[loop=above] node[above] {$x=0,
          \{x\}$} (1_xeqz_xyz);
        \draw (1_xeqz_xyz) edge node[right]{$\{y\}$}
        (2_xeqz_xyz);
        \draw (1_xeqz_clear) edge node[below]{$\{y\}$}
        (2_xeqz_y);
        \draw (2_xeqz_xyz) edge node[left]{$x \ge 1,\{z\}$}
        (3_xmzge1_xyz);
        \draw (2_xeqz_y) edge node[left]{$x \ge 1,\{z\}$}
        (3_xmzge1_yz);
        \draw (2_xeqz_clear) edge node[right]{$x \ge 1,\{z\}$}
        (3_xmzge1_z);
        \draw (3_xmzge1_xyz) edge[bend left] node[right]{$z=0$}
        (2_xmzge1_xyz);
        \draw (3_xmzge1_yz) edge[bend left] node[right]{$z=0$}
        (2_xmzge1_yz);
        \draw (3_xmzge1_z) edge[bend left] node[right]{$z=0$}
        (2_xmzge1_z);
        \draw (2_xmzge1_xyz) edge[bend left] node[left]{$x \ge
          1, \{z\}$} (3_xmzge1_xyz);
        \draw (2_xmzge1_yz) edge[bend left] node[left]{$x \ge
          1, \{z\}$} (3_xmzge1_yz);
        \draw (2_xmzge1_z) edge[bend left] node[left]{$x \ge
          1, \{z\}$} (3_xmzge1_z);
        \draw (2_xmzge1_clear) edge node[right]{$x \ge
          1, \{z\}$} (3_xmzge1_z);
        \begin{scope}[dashed]
          \draw (1_xeqz_xyz) edge node[above] {$\tau$} (1_xeqz_clear);
          \draw (2_xeqz_xyz) edge[bend left=20] node[above] {$\tau$}
          (2_xeqz_clear);
          \draw (2_xeqz_y) edge node[above]{$\tau$} (2_xeqz_clear);
          \draw (3_xmzge1_xyz) edge[bend left=14] node[above]{$\tau$}
          (3_xmzge1_clear);
          \draw (3_xmzge1_yz) edge[bend left=10] node[above]{$\tau$}
          (3_xmzge1_clear);
          \draw (3_xmzge1_z) edge node[above]{$\tau$} (3_xmzge1_clear);
          \draw (2_xmzge1_xyz) edge[bend right=15] node[below]{$\tau$}
          (2_xmzge1_clear);
          \draw (2_xmzge1_yz) edge[bend right=10] node[below]{$\tau$}
          (2_xmzge1_clear);
          \draw (2_xmzge1_z) edge node[above]{$\tau$} (2_xmzge1_clear);
        \end{scope}
      \end{scope}
    \end{tikzpicture}
  \end{center}
  \caption{A timed automaton $\Aa_1$ (top), its zone graph
    $\SG^\abs(\Aa_1)$ (middle) and the reachable part of the guessing
    zone graph $GZG^\abs(\Aa_1)$ (bottom) with $\tau$ self-loops
    omitted for clarity.}
  \label{fig:A1_GZG}
\end{figure}


%% file: fig/rgzg.tex
\begin{figure}[tp]
  \begin{center}
    \begin{tikzpicture}[font=\footnotesize]
      \begin{scope}
        \tikzstyle{every node}=[rectangle,draw]
        \node (1_xeqz_xz) at (0,0) {$1, x=z, \{x,z\}$};
        \node (1_xeqz_clear) at (12,0) {$1, x=z, \emptyset$};
        \node (2_xeqz_xz) at (0,-2) {$2, x=z, \{x,z\}$};
        \node (2_xeqz_clear) at (12,-2) {$2, x=z, \emptyset$};
        \node (3_xmzge1_z) at (6,-4) {$3,x-z \ge 1, \{z\}$};
        \node (3_xmzge1_clear) at (12,-4) {$3,x-z \ge 1, \emptyset$};
        \node (2_xmzge1_z) at (6,-6) {$2,x-z \ge 1, \{z\}$};
        \node (2_xmzge1_clear) at (12,-6) {$2,x-z \ge 1, \emptyset$};
      \end{scope}
      \begin{scope}[->,line width=1pt]
        \draw (-1.8,0) -- (1_xeqz_xz);
        \draw (1_xeqz_xz) edge[loop=above] node[above] {$x=0,
          \{x\}$} (1_xeqz_xz);
        \draw (1_xeqz_xz) edge node[right]{$\{y\}$}
        (2_xeqz_xz);
        \draw (1_xeqz_clear) edge node[left]{$\{y\}$}
        (2_xeqz_clear);
        \draw (2_xeqz_clear) edge node[right]{$x \ge 1,\{z\}$}
        (3_xmzge1_z);
        \draw (2_xeqz_xz) edge node[left]{$x \ge 1,\{z\}$}
        (3_xmzge1_z);
        \draw (3_xmzge1_z) edge[bend left] node[right]{$z=0$}
        (2_xmzge1_z);
        \draw (2_xmzge1_z) edge[bend left] node[left]{$x \ge 1,
          \{z\}$} (3_xmzge1_z);
        \draw (2_xmzge1_clear) edge node[right]{$x \ge 1, \{z\}$}
        (3_xmzge1_z);
        \begin{scope}[dashed]
          \draw (1_xeqz_xz) edge node[above]{$\tau$}
          (1_xeqz_clear);
          \draw (2_xeqz_xz) edge node[above]{$\tau$}
          (2_xeqz_clear);
          \draw (3_xmzge1_z) edge node[above]{$\tau$}
          (3_xmzge1_clear);
          \draw (2_xmzge1_z) edge node[above]{$\tau$}
          (2_xmzge1_clear);
        \end{scope}
      \end{scope}
    \end{tikzpicture}
  \end{center}
  \caption{The reachable part of the reduced guessing zone graph
    $\RGZG^\abs(\Aa_1)$ of automaton $\Aa_1$ in
    Figure~\ref{fig:A1_GZG} (with $\tau$ self-loops omitted for
    clarity).}
  \label{fig:rgzg}
\end{figure}


%% file: fig/zeno.tex
\begin{figure}[t]
  \centering
  \begin{tikzpicture}[shorten >=1pt,node distance=2cm,on
    grid,auto,font=\footnotesize]
  
    \node[state, initial, initial text={}] (q0) {$q_0$};
    \node[state] (q1) [right=of q0] {$q_1$};
    \node[state] (q2) [right=of q1] {$q_2$};
    \node[state] (q3) [right=of q2] {$q_3$};
    \node[state] (r0) [right=of q3] {$r_0$};
    \node[state,node distance=2.5cm] (r1) [right=of r0] {$r_1$};
    \node[state,node distance=2.5cm] (r2) [right=of r1] {$r_2$};
    
    \draw[->] (q0) edge [bend left] node {$\{x_1\}$}
    (q1); 
    \draw[->] (q0) edge [bend right] node[anchor=north]
    {$\{\bar{x_1}\}$} (q1);
   
    \draw[->] (q1) edge [bend left] node {$\{x_2\}$}
    (q2);
    \draw[->] (q1) edge [bend right] node[anchor=north]
    {$\{\bar{x_2}\}$} (q2);

    \draw[->] (q2) edge [bend left] node {$\{x_3\}$}
    (q3); 
    \draw[->] (q2) edge [bend right] node[anchor=north]
    {$\{\bar{x_3}\}$} (q3);

    \draw[->] (q3) -- (r0);

    \draw[->] (r0) edge [bend left=45] node[anchor=south] {$\bar{x}_1\ge
      1$} (r1);
    \draw[->] (r0) edge node[anchor=south] {$x_2\ge 1$}
    (r1);
    \draw[->] (r0) edge [bend right] node[anchor=north] {$\bar{x}_3\ge
      1$} (r1);

    \draw[->] (r1) edge [bend left=45] node[anchor=south] {$x_1\ge 1$}
    (r2);
    \draw[->] (r1) edge node[anchor=south] {$\bar{x}_2\ge 1$} (r2);
    \draw[->] (r1) edge [bend right] node[anchor=north] {$\bar{x}_3\ge
      1$} (r2);

    \draw[->] (r2) -- ++(0cm,-1.5cm) -- ([yshift=-1.5cm] q0.center) --
    (q0);
 \end{tikzpicture}
 \caption{$\AZ_\phi$ for $\phi = (p_1 \vee \neg p_2 \vee p_3) \wedge
   (\neg p_1 \vee p_2 \vee p_3)$}
\label{fig:zeno-eg}
\end{figure}

%% file: fig/ta-lift-safe-Z.tex
\begin{tikzpicture}[font=\footnotesize]
  \begin{scope}
    \tikzstyle{every node}=[circle,draw,minimum size=5mm]
    \node (s0) at (0,0) {$q_0$};
    \node (s1) at (0,-2) {$q_1$};
  \end{scope}
  \begin{scope}[->,line width=1pt]
    \draw (-0.7,0) -- (s0);
    \draw (s0) edge[in=60,out=120,loop] (s0);
    \draw (s0) edge[bend right=15] node[left] {$\{x\}$} (s1);
    \draw (s1) edge[bend right=15] node[right] {$x \geq 1$} (s0);
  \end{scope}
\end{tikzpicture}


%% file: fig/szg-ta-lift-safe-Z.tex
\begin{tikzpicture}[font=\footnotesize]
  \begin{scope}
    \tikzstyle{every node}=[draw]
    \node (s0x0f) at (0,0) {$q_0, \ x\geq 0, \ \free$};
    \node (s1x0f) at (0,-1.5) {$q_1, \ x\geq 0, \ \free$};
    \node (s0x1f) at (0,-3) {$q_0, \ x\geq 1, \ \free$};
    \node (s0x0s) at (5,0) {$q_0, \ x\geq 0, \ \slow$};
    \node (s1x0s) at (5,-1.5) {$q_1, \ x\geq 0, \ \slow$};
    \node (s0x1s) at (5,-3) {$q_0, \ x\geq 1, \ \slow$};
  \end{scope}
  \begin{scope}[->,line width=1pt]
    \draw (0,0.7) -- (s0x0f);
    \draw (s0x0f) edge[loop left] (s0x0f);
    \draw (s0x0f) -- node[left] {$\{x\}$} (s1x0f);
    \draw (s1x0f) edge[bend right] node[left] {$x\geq 1$} (s0x1f);
    \draw (s0x1f) edge[loop left] (s0x1f);
    \draw (s0x1f) edge[bend right] node[right] {$\{x\}$} (s1x0f);
    \draw (s0x0s) edge[loop right] (s0x0s);
    \draw (s0x0s) -- node[left] {$\{x\}$} (s1x0s);
    \draw (s1x0s) edge[bend right] node[left] {$x\geq 1$} (s0x1s);
    \draw (s0x1s) edge[loop right] (s0x1s);
    \begin{scope}[dashed]
      \draw (s0x0f) edge node[above] {$\tau$} (s0x0s);
      \draw (s1x0f) edge node[above] {$\tau$} (s1x0s);
      \draw (s0x1f) edge node[above] {$\tau$} (s0x1s);
    \end{scope}
  \end{scope}
  \begin{scope}[dotted,rounded corners]
    \draw (-2.5,1.25) rectangle (1.75,-3.75);
    \draw (0,-3.75) node[below] {\textbf{free mode}};
    \draw (3.25,1.25) rectangle (7.5,-3.75);
    \draw (5,-3.75) node[below] {\textbf{slow mode}};
  \end{scope}
\end{tikzpicture}


%% file: fig/pspace-widget.tex
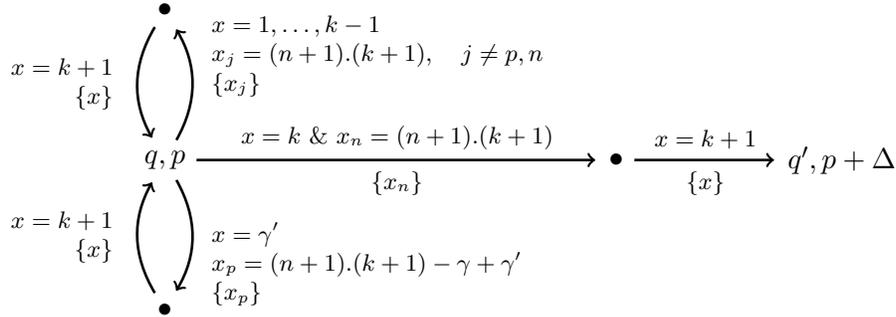
\begin{figure}[t]
  \centering
  \begin{tikzpicture}[shorten >=1pt,node distance=2cm,on
    grid,auto,->,line width=1pt]
  
    \node (qp) {$q,p$};
    \node (i) [above=of qp] {$\bullet$};
    \node (p) [below=of qp] {$\bullet$};
    \node[node distance=6cm] (n) [right=of qp] {$\bullet$};
    \node[node distance=3cm] (q'p') [right=of n] {$q',p+\D$};

    \draw (qp) edge [bend right] node[right,near end] {\footnotesize
      $\begin{array}{l}
        x=1,\dots,k-1\\
        x_j=(n+1).(k+1),\quad j \neq p,n\\
        \{x_j\}\\
      \end{array}$}
    (i);

    \draw (i) edge [bend right] node[left] {\footnotesize
      $\begin{array}{r}
        x=k+1\\
        \{x\}\\
      \end{array}$}
    (qp);

    \draw (qp) edge [bend left] node[right,near end] {\footnotesize
      $\begin{array}{l}
        x=\g'\\
        x_p=(n+1).(k+1)-\g+\g'\\
        \{x_p\}\\
      \end{array}$}
    (p);

    \draw (p) edge [bend left] node [left] {\footnotesize
      $\begin{array}{r}
        x=k+1\\
        \{x\}\\
      \end{array}$}
    (qp);

    \draw (qp) edge node[above] {\footnotesize
      $x=k \ \& \ x_n=(n+1).(k+1)$}
    node[below] {\footnotesize $\{x_n\}$}
    (n);

    \draw (n) edge node[above] {\footnotesize
      $x=k+1$}
    node[below] {\footnotesize $\{x\}$}
    (q'p');
 \end{tikzpicture}
 \caption{Widget for transition $(q,\g,\g',\D,q')$ on state $(q,p)$.}
 \label{fig:pspace-widget}
\end{figure}
